\documentclass[journal,twoside,web]{ieeecolor}
\usepackage{generic}
\usepackage{cite}
\usepackage{amsmath,amssymb,amsfonts}
\usepackage{graphicx}
\usepackage[caption=false]{subfig}
\usepackage{hyperref}
\hypersetup{hidelinks=true}
\usepackage{textcomp}
\def\BibTeX{{\rm B\kern-.05em{\sc i\kern-.025em b}\kern-.08em
    T\kern-.1667em\lower.7ex\hbox{E}\kern-.125emX}}
\usepackage{edgeStd}
\usepackage{empheq}
\usepackage[export]{adjustbox}

\usepackage{xcolor}
\allowdisplaybreaks
\begin{document}
\title{\LARGE \bf
Provably Efficient Sensor Allocation for Unknown High-dimensional Systems with Limited Sensing
}
\author{Yuyang Zhang$^{1,2}$, Derya Cansever$^{1}$ and Na Li$^{1}$
\thanks{This work is supported by NSF AI institute 2112085, NSF
ECCS 2328241, NIH R01LM014465. Yuyang Zhang is supported by the Kempner Graduate Fellowship.}
\thanks{$^{1}$Yuyang Zhang, Derya Cansever, and Na Li are with SEAS, Harvard University, USA.
{\tt\small \{yuyangzhang@g, derya\_cansever@fas, nali@seas\}.harvard.edu.}
}%
\thanks{$^{2}$Yuyang Zhang is also with the Kempner Institute, Harvard University.
}%
}
\maketitle

\begin{abstract}
    This paper focuses on learning efficient sensor allocations that ensure observability of unknown high-dimensional linear systems using only a small number of sensors. Existing methods either require an impractically large number of sensors or assume access to an observable allocation in advance. We propose a two-stage framework that overcomes these limitations: first, a novel system identification algorithm integrates information from multiple trajectories, each observing different subsets of state coordinates; then, a classic sensor allocation method is adapted to operate on the learned system parameters. Our non-asymptotic guarantees show that
    the proposed approach learns a sensor allocation with a near-optimal number of sensors when sensors can be allocated on any state coordinate. We further extend the results to settings with inaccessible state coordinates that are unavailable for sensor allocation.
\end{abstract}


\section{Introduction}\label{sec:intro}
High-dimensional dynamical systems are ubiquitous in real-world applications, such as power grids \cite{app_powergrid_1}, weather forecasting \cite{app_weather_1}, smart buildings \cite{app_building_1}, and neuroscience \cite{app_neuro_1}. Continuous monitoring of system states is essential for understanding their evolution, forecasting future trajectories, enabling adaptive control \cite{feedback_1}, and detecting potential failures \cite{failure_1}. Yet, monitoring all state variables is often impractical and costly, especially when sensors are expensive or energy-intensive \cite{sensor_1}. These challenges motivate the design of algorithms that learn efficient sensor allocations, ensuring observability with only a small number of sensors. 

In high-dimensional linear systems, classic sensor allocation methods typically assume access to an exact dynamical model \cite{sensor_1knownA, olshevsky2014minimal, clark2017submodularity} in order to compute efficient allocations. However, constructing accurate models from first principles is infeasible in such settings, making it necessary to estimate the model from data. This raises two major challenges. First, it is unclear whether the performance guarantees of the existing sensor allocation methods remain valid when applied to approximate models. Second, even estimating an approximate model is itself nontrivial. Standard least-squares methods require access to full state trajectories, which demand prohibitively extensive sensing \cite{sysid_2, sysid_1, sysid_11, sysid_nonlinear2, sysid_nonlinear3, sysid_nonlinear4}. Alternatively, Ho-Kalman-type algorithms can operate on partially observed trajectories \cite{sysid_5, proof_2tsiamis2019finite, sysid_13}, but these methods assume access to a sensor allocation that already guarantees observability. Without prior system knowledge, such an allocation is nearly impossible to obtain unless full-state measurements are available. 

To address the aforementioned difficulties, we propose a two-stage framework for learning efficient sensor allocations that guarantee observability with only a small number of sensors. \textit{In stage one}, we develop  a system identification (SYSID) algorithm that estimates the system dynamics without requiring observability and with an arbitrary number of sensors. Specifically, consider the following $r$-dimensional linear dynamical system $x_{t+1}=Ax_t+Bu_t+w_t$ with state $x_t\in\mathbb{R}^r$, input $u_t\in\mathbb{R}^m$, and noise $w_t\in\bbR^r$. The algorithm collects multiple data trajectories, each measuring a possibly different subset of the state coordinates, and jointly analyzes them to estimate the system matrices $A$ and $B$. 
We prove that the estimation error is bounded by $\tilde{\calO}(\sqrt{mr/(s_{\min}T)})$, where $\tilde{\calO}(\cdot)$ hides constants and logarithmic factors, $T$ is the trajectory length, and $s_{\min}$ is the least number of trajectories used to measure a state coordinate.
Intuitively, as long as every state coordinate is observed in at least one trajectory, i.e. $s_{\min}>0$, the algorithm accurately recovers the system matrices.
We highlight that this condition requires no prior structural knowledge and is easily satisfied---for instance, even with only one sensor, the system can be learned  by collecting $r$ data trajectories, each measuring a different state coordinate.

\textit{In stage two}, we adapt a greedy sensor allocation algorithm to operate on the estimated system parameters from stage one. A carefully designed rank estimation subroutine mitigates the influence of estimation noise, and the resulting algorithm is guaranteed to produce a sensor allocation ensuring observability with at most $\b{1+\log(r)}n^*$ sensors, where $n^*$ is the minimum number needed.
Unlike prior work \cite{sensor_1knownA, sensor_2knownA, olshevsky2014minimal} that requires exact system dynamics, our algorithm operates directly on models learned from data. 

Finally, we extend the proposed algorithms to systems with inaccessible state coordinates, where sensors can not be allocated. 
We show that the above algorithms, with minor modifications, can still learn efficient sensor allocations in this scenario, provided that the accessible state coordinates are sufficient to ensure observability of the system.

\textit{Notations:} For positive integer $a$, let $[a]$ denote set $\{1, 2, \cdots, a\}$. For $i\in[a]$, let $e_i\in\bbR^a$ denote the $a$-dimensional \textit{one-hot vector} with a $1$ on the $i$-th coordinate. We also let $e_0=e_a$ and $e_i=e_{i\bmod a}$ for any integer $i > a$, so that $e_{a+1}=e_{1}$, $e_{a+2}=e_{2}$, and so on.
For matrix $M$, we let $\sigma_{\min}(M)$ denote its smallest \textit{non-zero} singular value. We let $[M]_i$ denote its $i$-th row, $[M]^j$ denote its $j$-th column, and $[M]_i^j$ denote its $(i,j)$-th element. For set $\calI=\{i_1, \cdots, i_{|\calI|}\}$, we let $[M]^{\calI}=\begin{bmatrix}
    [M]^{i_1} & [M]^{i_2} & \cdots & [M]^{i_{|\calI|}}
\end{bmatrix}$ and let $[M]_{\calI}=\begin{bmatrix}
    [M]_{i_1}\t & [M]_{i_2}\t & \cdots & [M]_{i_{|\calI|}}\t
\end{bmatrix}\t$. We use $\tilde{\calO}(\cdot)$ to hide constants and logarithmic factors.

\subsection{Related Work}
\textbf{Sensor allocation algorithms.}
Classical sensor allocation algorithms assume access to a known system model \cite{sensor_1knownA, sensor_2knownA, clark2017submodularity, summers2015submodularity, olshevsky2014minimal}. These methods leverage the system model to evaluate observability-related objectives, and their theoretical guarantee relies heavily on the exactness of the system model. Whether such guarantees extend to approximate models remains unclear.
Several algorithms have been proposed to operate on approximate system models \cite{rw_approx, rw_approx2}, but their correctness is not guaranteed.
A recent line of work addresses data-driven sensor allocation for unknown dynamical systems \cite{rw_fullstate1, rw_fullstate2, rw_observable1}. However, these methods either require full-state measurements or assume an observable sensor allocation during the learning phase.



\textbf{System identification algorithms.}
System identification for linear dynamical systems has a rich history \cite{ljung1998system, keesman2011system}. Recent works have established non-asymptotic learning guarantees \cite{tsiamis2023statistical, sysid_9}. With full-state observations, ordinary least-squares methods achieve near-optimal rates for stable systems \cite{sysid_1, sysid_11} and have been extended to unstable settings \cite{sysid_2, sysid_3}.
For partially observed systems, Ho-Kalman-type algorithms recover system matrices from input-output data up to a similarity transformation \cite{sysid_5, proof_2tsiamis2019finite, sysid_13}. However, these methods require the system being observable, which is difficult to ensure without prior knowledge of the system. 
System identification has also been extended to nonlinear dynamics \cite{sysid_nonlinear1, sysid_nonlinear2, sysid_nonlinear3, brunton2016discovering}, though these works similarly require full state observations.

\textbf{Applications of sensor allocation.}
The call for efficient sensor allocation naturally arises in many large-scale systems. In power grids, phasor measurement units (PMUs) provide essential information for monitoring electricity activities, but their cost limits deployment \cite{app_powergrid_1, sensor_1}. In neuroscience, recording devices such as Neuropixels probes \cite{sensor_2} can monitor thousands of neurons simultaneously, yet the brain contains billions of neurons, making probe placement critical for understanding population-level dynamics \cite{app_neuro_1,app_neuro_3}.  
In weather forecasting, it is also important to allocate a limited number of observation stations to capture the dynamics of extremely high-dimensional atmospheric models \cite{app_weather_1, app_weather_2}. 

\section[]{Preliminaries \& Problem Setup}\label{sec:prelim}
\subsection{System \& Measurement Model}
Consider the following linear dynamical system
\begin{equation}\begin{split}\label{eq:sys}
    x_{t+1} = {}& Ax_{t} + Bu_{t} + w_{t}\in\bbR^r,
\end{split}\end{equation}
where $A$ and $B$ are unknown system matrices, $t$ is the time step and $r\gg1$ is the state dimension. Here  $u_{t}\in\bbR^m, x_{t}\in\bbR^r, w_{t}\in\bbR^r$ denote the input, state and process noise, respectively. 
For simplicity, we assume $x_{0} = 0$. We also assume $u_{t}\overset{\text{i.i.d.}}{\sim}\calN(0,\sigma_u^2I)$, $w_{t}\overset{\text{i.i.d.}}{\sim}\calN(0,\sigma_w^2I)$.\footnote{Our analysis can be easily adapted to inputs and noises with subGaussian distributions and non-isotropic variances.} 
With these parameters, we denote the above system by $\calM=(A,B,\sigma_u^2,\sigma_w^2)$.

Throughout the paper, we consider a target system 
satisfying the following two assumptions.
\begin{assumption}\label{assmp:stable}
    There exist positive constants $\rho_A < 1, \psi_A\geq 1$ such that $\norm{A^t} \leq \psi_A\rho_A^{t-1}, \quad \forall t\geq 0$.\qed
\end{assumption}
\noindent This assumption is standard in related literature \cite{sysid_5, zhang2024learning}. Intuitively, it assumes the $A$ matrix is ``stable'' so that states do not blow up over a long time period. 
    
\begin{assumption}\label{assmp:ctrb}
    $(A,B)$ is controllable, or equivalently, the following controllability matrix is full-row-rank
    \begin{equation*}
    \pushQED{\qed} 
        \begin{bmatrix}
        B & AB & \cdots & A^{r-1}B
    \end{bmatrix}.\qedhere
    \popQED
    \end{equation*}
\end{assumption}

We now describe how the system states $x_t$ are measured by sensors. Each sensor may be placed on one state coordinate for measurement. With $\barn$ sensors on coordinates $i_1, \cdots, i_{\barn}$, the corresponding measurement $y_t$ is given by
\begin{equation}\begin{split}\label{eq:measurement}
    y_{t} = \begin{bmatrix}
        e_{i_1}\t \\
        \vdots \\
        e_{i_{\barn}}\t
    \end{bmatrix} \b{x_t + \eta_t} = \begin{bmatrix}
        [x_t]_{i_1} +[\eta_t]_{i_1}\\
        \vdots \\
        [x_t]_{i_{\barn}}+[\eta_t]_{i_{\barn}}
    \end{bmatrix} \in\bbR^{\barn}.
\end{split}\end{equation}
Here $\eta_{t}$ denotes the measurement noise and we assume $\eta_{t}\overset{\text{i.i.d.}}{\sim}\calN(0,\sigma_{\eta}^2I)$. In matrix $[e_{i_1}\t,\cdots,e_{i_{\barn}}\t]\t$, each row is a distinct one-hot vector $e_i\t$, representing a sensor on the $i$-th state coordinate. We call such matrices, whose rows are distinct one-hot vectors, \textit{measurement matrices}. For any measurement matrix $M$, we also define its observability matrix 
\begin{equation*}\begin{split}
    O(M) \coloneqq \begin{bmatrix}
        M\t & (M A)\t & \cdots & (M A^{r-1})\t
    \end{bmatrix}\t.
\end{split}\end{equation*}

\subsection{The Sensor Allocation Problem}
As discussed in Section~\ref{sec:intro}, this paper focuses on unknown high-dimensional systems with a large state dimension $r$. The goal is to learn a sensor allocation that ensures observability with a small number of sensors. 
Since the system is unknown, a natural approach is to first identify the system from data and design the sensor allocation based on the estimated model. However, standard identification methods either require a massive number of sensors to cover all state coordinates or a sensor allocation that renders the system observable, both of which are impractical in practice.

To overcome this challenge, we propose the following problem setting where $K$ data trajectories are collected, each measuring a possibly different set of state coordinates. Specifically, let $\barn$ denote the number of available sensors during data collection. For each trajectory $k\in[K]$, we choose a measurement matrix $C_k\in\bbR^{\barn\times r}$ with $\barn$ sensors, inject inputs $\calU_k=\{u_{k,t}\}_{t\in[0,T]}$, and observe $\calY_k=\{y_{k,t}\}_{t\in[0,T+1]}$ as specified by \Cref{eq:measurement}.

Our goal is to choose measurement matrices $\{C_k\}_{k\in[K]}$, collect the datasets $\bigcup_{k\in[K]} \b{\calY_k\cup\calU_k}$, and learn a single measurement matrix $\htC$
such that $(A,\htC)$ is observable, or equivalently, the observability matrix $O(\htC)$ is full-column-rank. To avoid confusion, we will refer to the final output $\htC$ as the \textit{sensor allocation} to differentiate it from the \textit{measurement matrices} $\{C_k\}_{k\in[K]}$ used only during data collection.



\subsection{Overview of Our Solution}
In the following sections, we propose a two-stage framework to solve the problem. \textit{In stage one (\Cref{sec:sysid})}, we develop a novel SYSID algorithm that dynamically selects different state coordinates to measure in different state trajectories and jointly analyzes the multi-trajectory dataset. It estimates system matrices $A$ and $B$ accurately with any number of sensors $\barn\geq 1$ and without observability requirement. \textit{In stage two (\Cref{sec:placement})}, we design a data-driven sensor allocation algorithm that operates on the estimated system model in the first stage and outputs a near-optimal sensor allocation. 

 
\section{Stage One: SYSID with Dynamic Measurement Selection}\label{sec:sysid}
In this section, we present the SYSID algorithm outlined in \Cref{sec:prelim}. We discuss the rationale of the algorithm design and provide its theoretical guarantee. 
Results in this section may be of independent interest for learning high-dimensional linear systems with a limited number of sensors.

\subsection{Rationale of the algorithm design}\label{subsc:rationale}
We first state our principle to choose the measurement matrices $\{C_k\}_{k\in[K]}$: \textit{measure every state coordinate in at least one trajectory}.
As we will show later in \Cref{subsc:theory}, this principle guarantees the accurate recovery of system matrices $A$ and $B$. 
One can easily follow this principle for any given number of sensors $\barn$ with measurement matrices
\begin{equation*}\begin{split}
    C_k = {}& \begin{bmatrix}
        e_{(k-1)\barn+1}\t & e_{(k-1)\barn+2}\t & \cdots & e_{k\barn}\t
    \end{bmatrix}\t, ~\forall k\in[K],
\end{split}\end{equation*}
where $K = \lceil sr/\barn\rceil$ for any chosen positive integer $s\geq 1$. Here one-hot vector $e_i$ is defined cyclically (\Cref{sec:intro}, Notations) such that $e_{r+1}=e_1$, $e_{r+2}=e_2$, and so on. Intuitively, $\{C_k\}_{k\in[K]}$ cyclically allocate sensors on state coordinates so that each one of them is measured in at least $s$ trajectories.

With the chosen measurement matrices and the corresponding dataset $\bigcup_{k\in[K]} \b{\calY_k\cup\calU_k}$, we aim to learn the system matrices $A,B$. To motivate the learning algorithm, we consider the simple example where $K=r$ and $C_k = e_k\t (\forall k\in[r])$, i.e., the $k$-th trajectory measures the $k$-th state coordinate $\bs{x_{k,t}}_k$. 
To learn the system matrices, one naive approach is to perform least-squares on pairs $({\bs{x_{k,t+1}+\eta_{k,t+1}}_k, \bs{x_{k,t}+\eta_{k,t}}_k, u_{k,t}})$.
However, from the system dynamics $x_{k,t+1} = A x_{k,t} + B u_{k,t} + w_{k,t}$, we notice that $[x_{k,t+1}]_k$ only depends on $[x_{k,t}]_k$ through the $(k,k)$-th element of $A$. Therefore, this naive approach can recover no more than the diagonal elements of $A$.

Our algorithm follows a different approach. By a recursive expansion of the state dynamics (\Cref{eq:sys}),
\begin{align*}
    {}&  x_{k,t+1} = \sum_{\tau=0}^d A^{\tau}Bu_{k,t-\tau}\\
    + {}& \underbrace{\sum_{\tau=d+1}^{t} A^{\tau}Bu_{k,t-\tau} + \sum_{\tau=0}^t A^{\tau}w_{k,t-\tau}}_{\calZ_{k,t}}.
\end{align*}
Here $d$ is a hyperparameter to be chosen, and the last two terms are residuals and noises, denoted by $\calZ_{k,t}$. We write the above equation in the matrix form
\vspace{-1em}\begin{equation}\begin{split}\label{eq:alg_1}
    x_{k,t+1} = \underbrace{\begin{bmatrix}
        B & AB & \cdots & A^{d}B
    \end{bmatrix}}_{G}\underbrace{\begin{bmatrix}
        u_{k,t}\\
        \vdots\\
        u_{k,t-d}\\
    \end{bmatrix}}_{U_{k,t}(d)} + \calZ_{k,t},
\end{split}\end{equation}
and observe that state $x_{k,t+1}$ is a linear combination of the past inputs $u_{k,t-d:t}$, denoted by $U_{k,t}(d)$, with coefficients $G$. The columns of $G$, i.e. $B, AB, \cdots, A^dB$, are often referred to as the ``Markov parameters''. Therefore
\begin{equation*}\begin{split}
    y_{k,t+1} ={}&  C_k\b{x_{k,t+1} + \eta_{k,t+1}}\\
    = {}& C_kG U_{k,t}(d) + C_k\b{\calZ_{k,t}+\eta_{k,t+1}}.    
\end{split}\end{equation*}
Since the input $U_{k,t}(d)$, observation $y_{k,t+1}$ and measurement matrix $C_k$ are known, we can perform least-squares and learn the $G$ matrix accurately. The system matrices $A,B$ can be subsequently recovered from $G$.

\subsection{Algorithm}
\vspace{-0.5em}\begin{algorithm}[h]
    \caption{SYSID with Dynamic Measurement Selection}\label{alg:alloc}
    \begin{algorithmic}[1]     
        \State \textbf{Init: } trajectory length $T$, estimation rank $d\geq r$, number of sensors $\barn$, repetition $s$;
        
        \State Select $K \gets \lceil sr/\barn\rceil$ and $C_k$, $\forall k\in[K]$\vspace{-0.5em}
        \begin{equation*}\begin{split}
            C_k \gets {}& \begin{bmatrix}
                e_{(k-1)\barn+1}\t & e_{(k-1)\barn+2}\t & \cdots & e_{k\barn}\t
            \end{bmatrix}\t;
        \end{split}\end{equation*}

        \State Collect samples from system $\calM$:
        \For{$k\in[K]$}
            \State Inject inputs $\{u_{k,t}\sim\calN(0,\sigma_u^2 I)\}_{t\in[0,T]}$ and observe $\{y_{k,t}\}_{t\in[0,T+1]}$ with observer $C_k$;
        \EndFor

        \State Approximate Markov parameters via least-squares:\vspace{-0.5em}
        \begin{equation*}\begin{split}
            \wh{G} \gets {}& \mathop{\arg\min}_{G'} \sum_{k\in[K]}\sum_{t=d}^{T}\norm{y_{k,t+1} - C_k G' U_{k,t}(d)}^2\\
            {}& \eqqcolon \begin{bmatrix}
                \wh{B} & \wh{AB} & \cdots & \wh{A^dB}
            \end{bmatrix};
        \end{split}\end{equation*}
        \State \textbf{Output: } $\wh{G}$.
    \end{algorithmic}
\end{algorithm}

\vspace{-1em}Following the above design idea, we now present \Cref{alg:alloc}. It starts by choosing a set of measurement matrices $\{C_k\}_{k\in[K]}$, each specifying the measured state coordinates in one trajectory (line 2). 
With the measurement matrices, $K$ observation trajectories $\{y_{k,t}\}_{k\in[K],t\in[0,T+1]}$ are collected with the corresponding Gaussian inputs $\{u_{k,t}\}_{k\in[K],t\in[0,T]}$ (lines 3-5). 
The $G$ matrix, or Markov parameters $B, AB, \cdots, A^dB$ as in \Cref{eq:alg_1}, is approximated via least-squares from this dataset (line 6). As will be justified below, we choose $d$ larger than $r$ so that the system matrices can be recovered from approximated Markov parameters $\wh{G}$. 

With output $\htG\eqqcolon \begin{bmatrix}
    \wh{B} & \wh{AB} & \cdots & \wh{A^dB}
\end{bmatrix}$, one can recover the system matrices by the following procedure:
\begin{subequations}\label{alg:recover}
\begin{empheq}[box=\fbox]{align}
    {}& \wh{G}^- = [\htG]^{1:md}, \quad \wh{G}^+ = [\htG]^{m+1:m(d+1)}.\\
    {}& \wh{A} = \wh{G}^+ (\wh{G}^-)^{\dagger}, \quad \wh{B} = \big[\wh{G}^-\big]^{1:m}.
\end{empheq}    
\end{subequations}
\noindent We select the first and last $md$ columns of $\wh{G}$ to construct $\wh{G}^-$ and $\wh{G}^+$, which approximate $G^- \coloneqq {\begin{bmatrix}
    B & AB & \cdots & A^{d-1}B
\end{bmatrix}}$ and $G^+ \coloneqq {\begin{bmatrix}
    AB & A^2B & \cdots & A^dB
\end{bmatrix}} = AG^-$, respectively. When the system is controllable (Assumption \ref{assmp:ctrb}) and $d\geq r$, $G^-$ is full-row-rank and thus right-invertible. Therefore, if the approximations are accurate enough, $\wh{G}^+(\wh{G}^-)^\dagger \approx G^+(G^-)^\dagger = AG^-(G^-)^\dagger = A$. The $B$ matrix is recovered by selecting the first $m$ columns of $\wh{G}^-$.

\subsection{Theoretical Guarantee}\label{subsc:theory}
We now establish the theoretical guarantee for \Cref{alg:alloc}. Although \Cref{alg:alloc} adopts a specific set of measurement matrices $\{C_k\}_{k\in[K]}$ in line 2, the theoretical result below holds for any measurement matrices.
\begin{theorem}[Error Bound for Markov Parameters]\label{thm:markov}
    Consider system $\calM=(A,B,\sigma_u^2,\sigma_w^2)$, which satisfies Assumption \ref{assmp:stable}, and observation noise covariance $\sigma_{\eta}^2I$. Let $\psi_B = \max\{\norm{B},1\}$. 
    Consider any measurement matrices $\{C_k\}_{k\in[K]}$.
    For matrix $\sum_{k\in[K]}C_k\t C_k$, let $\calI=\big\{i:[\sum_{k\in[K]}C_k\t C_k]_i^i \neq 0\big\}$ denote the indices of its non-zero diagonal elements, and let $s_{\max}$ and $s_{\min}$ denote the maximum and minimum of the non-zero diagonal elements, respectively.

    Consider any $\delta\in(0,1)$. If the estimation rank $d\geq r$ and the trajectory length $T$ satisfy\vspace{-0.5em}
    \begin{equation}\begin{split}\label{eq:markov_1}
        T \geq {}&  2d + 64c\cdot rmd\big(\log^4(2mdT)+\log(\frac{60K}{\delta})\big)
    \end{split}\end{equation}
    for absolute constant $c$ in \cite[Theorem 4.1 with $L=1$]{krahmer2014suprema},
    then with probability at least $1-\delta$, $\wh{G}$ from \Cref{alg:alloc} satisfies:
    \begin{equation*}\begin{split}
        \norm{\bs{\wh{G}-G}_{\calI}} \leq {}& \kappa_1\sqrt{\frac{s_{\max}}{s_{\min}}\frac{md}{s_{\min}T}}.
    \end{split}\end{equation*}
    Here $\kappa_1 = 24\sqrt{6\dfrac{(\sigma_u^2+\sigma_w^2+\sigma_{\eta}^2)\psi_A^2\psi_B^2}{\sigma_u^2\rho_A^2(1-\rho_A)^2}\log\b{\dfrac{120}{\delta}}}$, $G = \begin{bmatrix}
        B & \cdots & A^dB
    \end{bmatrix}$.\qed
\end{theorem}
In the above theorem, $C_k\t C_k$ is a diagonal matrix because $C_k$ is a measurement matrix whose rows are distinct one-hot vectors. The $i$-th diagonal element $[C_k\t C_k]_i^i$ equals $1$ if the $k$-th trajectory measures state coordinate $i$, and $0$ otherwise. 
Consequently, $[\sum_{k\in[K]}C_k\t C_k]_i^i$ counts the number of trajectories that measure coordinate $i$.

Intuitively, the theorem states that if we measure state coordinates $\calI$ across the trajectories, we can accurately learn the corresponding rows $[G]_{\calI}$ of the Markov parameters.
Under the simple measurement strategy (line 2 in \Cref{alg:alloc}), every state coordinate is measured in at least $s$ and at most $s+1$ trajectories, leading to $\calI=[r]$ and $s\leq s_{\min}\leq s_{\max}\leq s+1$. The error of $\wh{G}$ then simplifies to $\tilde{\calO}\b{\sqrt{md/(sT)}}$, which decays to $0$ as the total number of samples $sT$ grows.
We defer the detailed proof to \Cref{proof:thm1}.

Now we take a step further to bound the errors of $\htA, \htB$ recovered from $\wh{G}$ (\Cref{alg:recover}):
\begin{lemma}[Error Bound for System Matrices]\label{lem:param}
    Consider the setting of \Cref{thm:markov}. Suppose the system satisfies Assumption \ref{assmp:ctrb}. Suppose $\sum_{k\in[K]}C_k\t C_k$ is invertible.
    If 
    \begin{equation}\begin{split}\label{eq:param_2}
        \norm{\wh{G}-G} \leq \sigma_{\min}\b{G^-}/2,
    \end{split}\end{equation}
    then the outputs $\wh{G}$ from \Cref{alg:alloc} and $\wh{A},\wh{B}$ from \Cref{alg:recover} satisfy 
    \begin{equation*}\begin{split}
        \norm{\wh{A} - A} \leq \frac{6\norm{G}}{\sigma_{\min}^2\big(G^-\big)}\norm{\wh{G}-G}, \quad 
        \norm{\wh{B}-B} \leq {}& \norm{\wh{G}-G}.
    \end{split}\end{equation*}
    Here $G = \begin{bmatrix}
        B & \cdots & A^{d}B
    \end{bmatrix}$, $G^- = \begin{bmatrix}
        B & \cdots & A^{d-1}B
    \end{bmatrix}$.\qed
\end{lemma}

By \Cref{thm:markov}, condition $\|\wh{G}-G\| \leq \sigma_{\min}\b{G^-}/2$ in the above lemma is readily satisfied when $s_{\min}T$ is large enough. Under this condition, $\wh{G}^-$ (\Cref{alg:recover}) has the same rank as $G^-$, which, together with the controllability assumption (Assumption \ref{assmp:ctrb}), ensures that $\wh{G}^-$ is right-invertible and $\wh{A}$ is well-defined. 

Combining the above result with \Cref{thm:markov}, we obtain $\norm{\wh{A}-A}, \norm{\wh{B}-B} \leq \tilde{\calO}\b{\sqrt{md/(sT)}}$ when $T$ is sufficiently large. 
This confirms that \Cref{alg:alloc} accurately recovers the system matrices.


\section{Stage Two: Data-Driven Sensor Allocation}\label{sec:placement}
With the approximated system matrix $\htA$ from Stage One, we now learn an efficient sensor allocation to make the system observable with a small number of sensors. 

We present a greedy algorithm (\Cref{alg:2}) that learns such a sensor allocation by iteratively maximizing the observability matrix rank. 
At each iteration, the algorithm first estimates the  observability matrix rank under the current sensor allocation $\htC$ using a rank estimation subroutine $\calE$.
If  $\htC$ does not make the system observable, adding a sensor on a new state coordinate (or equivalently, adding a new row to $\htC$) may increase the observability matrix rank. 
The algorithm estimates rank increase $h_i$ for every state coordinate $i$ not measured by the current $\htC$ (line 3-4) using the rank-estimation subroutine $\calE$, which estimates the observability matrix rank for any given measurement matrix $\htC$. It subsequently adds a sensor on the coordinate with the largest rank increase (line 5). The iteration ends when observability matrix rank reaches $r$.

\begin{algorithm}[ht]
    \caption{Sensor Allocation}
    \label{alg:2}
    \begin{algorithmic}[1]     
        \State \textbf{Init: } coordinate set $\calI = [r]$, initial allocation $\htC\in\bbR^{0\times r}$, rank estimation subroutine $\calE$;
        \While{$\calE(\htC)<r$}
            \State Estimate rank increase of observability matrix:
            \For{any state coordinate $i\in\calI$ that is not currently measured
            }
            \begin{equation*}\begin{split}
                h_{i} \gets \calE\b{
                \begin{bmatrix}
                    \htC\\
                    e_i\t
                \end{bmatrix}} - \calE\b{\htC};
            \end{split}\end{equation*}
            \EndFor
            \State Select state coordinate with the largest rank increase:
            \begin{equation*}\begin{split}
                i^* \gets {}& \arg\max_i h_{i}, \quad
                \htC \gets \begin{bmatrix}
                        \htC\\
                        e_{i^*}\t
                    \end{bmatrix}.
            \end{split}\end{equation*}
        \EndWhile

        \State \textbf{Output: } $\htC\in\bbR^{\htn\times r}$.
    \end{algorithmic}
\end{algorithm}

In \Cref{alg:2}, the key component is the rank estimation subroutine $\calE$. Unlike existing literature that either assumes this rank information as prior knowledge or requires the exact system matrix $A$ \cite{sensor_1knownA, sensor_2knownA, olshevsky2014minimal}, we design a subroutine $\calE$ (\Cref{alg:rank}) to estimate the rank from $\htA$, the data-driven estimate of $A$ from \Cref{alg:alloc}. 
\begin{algorithm}[ht]
    \caption{Rank Estimation Subroutine $\calE$}
    \label{alg:rank}
    \begin{algorithmic}[1]
    \State \textbf{Initialization:} receive $1/(s_{\min}T)$ and $\htA$.
    \State \textbf{Input:} measurement matrix $C$; 
    \State Form observability matrix:
    \begin{equation*}\begin{split}
        \wh{O}(C) \gets \begin{bmatrix}
            C\t & (C\htA)\t & \cdots & (C\htA^{r-1})\t
        \end{bmatrix}\t;
    \end{split}\end{equation*}
    \State Rank estimation:
    \begin{equation*}\begin{split}
        \htr\b{C} \gets \max \left\{i: \sigma_i\b{\wh{O}(C)} > \sqrt[4]{\frac{1}{s_{\min}T}}\right\};
    \end{split}\end{equation*}
    \State \textbf{Output:} $\htr\b{C}$.
    \end{algorithmic}
\end{algorithm}
Specifically, given $\htA$, the subroutine first forms an approximation of the true observability matrix $O(C)= \begin{bmatrix}
    C\t & (CA)\t & \cdots & (CA^{r-1})\t
\end{bmatrix}\t$ for the query matrix $C$, denoted by $\htO(C)$ (line 3). When $\htA$ is sufficiently accurate, the singular values of $\htO(C)$ naturally split into two groups: $\rank(O(C))$ large singular values corresponding to the non-zero singular values of $O(C)$, and $r-\rank(O(C))$ small singular values corresponding to the zero singular values of $O(C)$. The threshold $\sqrt[4]{1/(s_{\min}T)}$ is carefully chosen to separate the two groups of singular values. Here the constants $s_{\min},T$ are defined in \Cref{thm:markov}.

Now we state the performance guarantee for \Cref{alg:2}.
\begin{theorem}\label{thm:alloc}
    Consider system $(A,B,\sigma_u^2,\sigma_w^2)$ satisfying Assumption \ref{assmp:stable} with $\psi_A,\rho_A$ and Assumption \ref{assmp:ctrb}. Consider $\htA$ satisfying 
$\norm{\htA-A} \leq \kappa_2\sqrt{s_{\max}md/s_{\min}^2T}$ 
for some constant $\kappa_2\geq1$. Let $\sigma_O \coloneqq \min_{C\in\calC}\sigma_{\min}(O(C))$, where $\calC$ denotes the set of all non-zero measurement matrices.
If 
    \begin{equation}\begin{split}\label{eq:sub1}
        {}& T > \max\left\{\frac{16}{s_{\min}\sigma_O^4},\frac{10\psi_A^8\kappa_2^4}{\rho_A^8(1-\rho_A)^8}\cdot \frac{s_{\max}^2}{s_{\min}^2}\frac{m^2d^2}{s_{\min}}\right\},
    \end{split}\end{equation}
    then the output $\htr(C)$ of Algorithm \ref{alg:rank} is accurate for any measurement matrix $C$. Namely,
    \begin{equation}\begin{split}\label{eq:sub2}
        \htr(C) = {}& \rank\b{O(C)} = \rank(\begin{bmatrix}
            C\t & \cdots & (CA^{r-1})\t
        \end{bmatrix}\t).
    \end{split}\end{equation}
    Furthermore, with \Cref{alg:rank} as the rank estimation subroutine, the output $\htC\in\bbR^{\htn\times r}$ of \Cref{alg:2} satisfies
    \begin{equation}\begin{split}\label{eq:sub5}
        (A, \htC) \text{ is observable}, \quad \htn \leq \b{1+\log(r)}\cdot n^*.
    \end{split}\end{equation}
    Here $n^*$ is the least number of sensors needed to make the system observable. \qed
\end{theorem}
\noindent Combining this result with \Cref{thm:markov}, \Cref{alg:2} is guaranteed to learn an efficient sensor allocation that is optimal up to constant $1+\log(r)$. 
We remark that the constants in \eqref{eq:sub1} reflect worst-case analysis. Our numerical experiments in \Cref{sec:sim} confirm that the algorithm succeeds with a moderate trajectory length in practice.


\section[]{Proofs for \Cref{thm:markov}, \ref{thm:alloc}, and \Cref{lem:param}}\label{proof:thm1}
In this section, we provide proofs of the aforementioned main theoretical results. Proofs of auxiliary lemmas used in this section are deferred to the appendices.
\subsection[]{Proof of \Cref{thm:markov}}

\textit{Proof.} We begin by analyzing the least-squares problem (line 6 of \Cref{alg:alloc}) and deriving the optimality condition for its solution $\wh{G}$. For each row of $\wh{G}$, we express its error as the combination of a noise term and an estimated signal covariance term (\textit{Step 1}). We then derive an upper bound for the noise term and the error of the estimated signal covariance inverse term (\textit{Step 2}). Finally, we concatenate the rows of $\wh{G}$ and upper bound the error concatenation (\textit{Step 3}). 

\textit{Step 1.} Consider the least-squares objective in line 6 of \Cref{alg:alloc}. By taking the gradient with respect to $\wh{G}$ and setting it to zero, we get the optimality condition on $\wh{G}$:
\begin{equation}\begin{split}\label{eq:markov_text6}
    {}& \sum_{k\in[K], t\in[d,T]} C_k\t \b{C_k \wh{G} U_{k,t}(d)-y_{k,t+1}}U_{k,t}(d)\t = 0.
\end{split}\end{equation}
Here $U_{k,t}(d) = \begin{bmatrix}
    u_{k,t}\t & \cdots & u_{k,t-d}\t
\end{bmatrix}\t$.

Recall $y_{k,t+1} = C_k\b{x_{k,t+1} + \eta_{k,t+1}}$. 
We expand $x_{k,t+1}$ as $x_{k,t+1} =  \sum_{\tau=0}^{t} A^{\tau} \b{B u_{k,t-\tau} + w_{k,t-\tau}}$ by state dynamics \Cref{eq:sys} and obtain the following expansion for $y_{k,t+1}$

\begin{equation}\begin{split}\label{eq:markov_text7}
    y_{k,t+1} = {}& C_k\b{\sum_{\tau=0}^{t} A^{\tau} \b{B u_{k,t-\tau} + w_{k,t-\tau}}+\eta_{k,t+1}}\\
    = {}& C_kGU_{k,t}(d) + C_k\underbrace{\b{\eta_{k,t+1} + \sum_{\tau=0}^{d} A^{\tau} w_{k,t-\tau}}}_{\Delta_{k,t+1}^2}\\
    {}&  + C_k\underbrace{\sum_{\tau=d+1}^{t} A^{\tau} \b{B u_{k,t-\tau} + w_{k,t-\tau}}}_{\Delta_{k,t+1}^1}.
\end{split}\end{equation}

\noindent Here recall that $G=\begin{bmatrix}
    B & \cdots & A^dB
\end{bmatrix}$ is the true Markov parameters. 
Substituting into the optimality condition of $\wh{G}$ (\Cref{eq:markov_text6}) and rearranging the terms give
\begin{equation}\begin{split}\label{eq:markov_text}
    {}& \sum_{k\in[K], t\in[d,T]} C_k\t C_k \b{\wh{G} - G }U_{k,t}(d)U_{k,t}(d)\t\\
    = {}& \underbrace{\sum_{k\in[K], t\in[d,T]} C_k\t C_k \b{\Delta_{k,t+1}^1 + \Delta_{k,t+1}^2}U_{k,t}(d)\t}_{\Delta}.
\end{split}\end{equation}

Consider any coordinate $i\in\calI=\big\{i:[\sum_{k\in[K]}C_k\t C_k]_i^i \neq 0\big\}$ that is measured in at least one trajectory. To bound the error of $[\wh{G}]_i$, we take the $i$-th row of \Cref{eq:markov_text} on both sides and get
\begin{equation}\begin{split}\label{eq:text9}
    \sum_{k\in[K], t\in[d,T]} \bs{C_k\t C_k}_i\b{\wh{G} - G} U_{k,t}(d)U_{k,t}(d)\t = [\Delta]_i\\
\end{split}\end{equation}
Since $C_k$ is the measurement matrix for trajectory $k$, we know
\begin{equation*}\begin{split}
    [C_k\t C_k]_i = \left\{
    \begin{array}{ll}
        e_i\t, \quad & \text{state coordinate } i \text{ is observed}\\
        & \text{in trajectory } k,\\
        0, \quad & \text{otherwise}.
    \end{array}
    \right.
\end{split}\end{equation*}
Therefore, $[C_k\t C_k]_i(\wh{G}-G)$ equals $[\wh{G}-G]_i$ if state coordinate $i$ is measured in trajectory $k$, and equals $0$ otherwise. It then follows by \Cref{eq:text9} that
\begin{equation}\begin{split}\label{eq:markov_text3}
    {}& \bs{\wh{G}-G}_i \underbrace{\sum_{k\in\calK_i, t\in[d,T]} U_{k,t}(d)U_{k,t}(d)\t}_{\wh{\Sigma}_i} = [\Delta]_i.
\end{split}\end{equation}
Here we use $\calK_i\coloneqq \{k: [C_k\t C_k]_i^i\neq 0\}$ to denote the trajectories where the $i$-th state coordinate is measured, and use 
 $\wh{\Sigma}_i$ to denote the estimated signal covariance.

\textit{Step 2.} To bound the error $\bs{\wh{G}-G}_i$ , we first upper bound norm of the noise term $\Delta$ defined in \Cref{eq:markov_text}. This is challenging because $\Delta_{k,t}^1$, $\Delta_{k,t}^2$ and $U_{k,t}(d)$ are dependent across time steps $t$. 
To tackle this issue, we carefully manipulate the martingale formed by the terms and utilize the self-normalizing martingale analysis from previous work \cite{abbasi2011improved,sysid_2}, yielding the following lemma:
\begin{lemma}\label{lem:cross_all}
    Consider the setting of \Cref{thm:markov}. With probability at least $1-\delta/2$,
    \begin{equation*}\begin{split}
        \norm{\Delta} \leq \kappa_4\cdot \sqrt{mds_{\max}(T-d+1)},
    \end{split}\end{equation*}
    where $\kappa_4 = \kappa_1\cdot \sigma_u^2/4$.\qed     
\end{lemma}
\noindent Importantly, $\norm{\Delta}$ scales sublinearly with ${s_{\max}(T-d+1)}$, where $s_{\max} = \max_i|\calK_i|$ is the maximum diagonal element of $\sum_{k\in[K]}C_k\t C_k$. 

We then show that the estimated signal covariance $\wh{\Sigma}_i$ is invertible and upper bound the deviation of its inverse $\wh{\Sigma}_i^{-1}$ from the true covariance inverse $\Sigma_i^{-1} \coloneqq 1/\b{|\calK_i|(T-d+1)\sigma_u^2}I$.

By a covariance concentration bound (\Cref{lem:input_conc} in Appendix \ref{proof:thm}), we know that for $T$ satisfying \Cref{eq:markov_1}, with probability at least $1-\delta/(2K)$, the following concentration holds for any trajectory $k\in[K]$,
\begin{equation*}\begin{split}
    {}& \frac{1}{T-d+1} \sum_{t=d}^T U_{k,t}(d)U_{k,t}(d)\t \succeq \frac{\sigma_u^2}{2}I,\\
    {}& \norm{\frac{1}{T-d+1} \sum_{t=d}^T U_{k,t}(d)U_{k,t}(d)\t - \sigma_u^2I} \leq \kappa_5 \sqrt{\frac{md}{T}},
\end{split}\end{equation*}
where $\kappa_5 = 2\sigma_u^2\sqrt{c\b{\log^4(2mdT)+\log\b{2K/\delta}}}$ and $c$ is an absolute constant from previous work \cite[Theorem 4.1 with $L=1$]{krahmer2014suprema}.
Thus, with a union bound, the above holds for all trajectories $k\in[K]$ with probability at least $1-\delta/2$.
When the above events hold, it naturally follows by the definition of $\wh{\Sigma}_i$ that for all $i\in[r]$,
\begin{equation*}\begin{split}
    {}& \wh{\Sigma}_i = \sum_{k\in\calK_i, t\in[d,T]} U_{k,t}(d)U_{k,t}(d)\t \succeq \frac{\sigma_u^2}{2}|\calK_i|(T-d+1)I,\\
\end{split}\end{equation*}
This confirms the invertibility of $\wh{\Sigma}_i$. Moreover,
\begin{equation*}\begin{split}
        {}& \norm{\wh{\Sigma}_i - \Sigma_i} \leq \sum_{k\in\calK_i} \norm{\sum_{t=d}^T U_{k,t}(d)U_{k,t}(d)\t - \sigma_u^2(T-d+1)I}\\
    {}& \hspace{4.5em}\leq  |\calK_i|T\cdot \kappa_5 \sqrt{\frac{md}{T}}.
\end{split}\end{equation*}
Thus, with probability at least $1-\delta/2$, the following holds for all $i\in[r]$
\begin{equation}\begin{split}\label{eq:markov_10}
    {}& \norm{\wh{\Sigma}_i^{-1} - \Sigma_i^{-1}} \leq \norm{\wh{\Sigma}_i^{-1}}\norm{\Sigma_i - \wh{\Sigma}_i}\norm{\Sigma_i^{-1}}\\
    \leq {}& \b{\frac{2}{\sigma_u^2|\calK_i|(T-d+1)}} \cdot \b{|\calK_i|T\kappa_5 \sqrt{\frac{md}{T}}}\\
    {}& \hspace{2em} \cdot \b{\frac{1}{\sigma_u^2|\calK_i|(T-d+1)}}\\
    \leq {}& \frac{2\kappa_5}{\sigma_u^4(T-d+1)^2}\sqrt{\frac{mdT}{s_{\min}^2}}\\
\end{split}\end{equation}
Here in the last inequality, we have used $s_{\min}=\min_{i\in\calI}|\calK_i|$, which is the minimum nonzero diagonal element of $\sum_{k\in[K]}C_k\t C_k$.
This inequality shows that the distance between $\wh{\Sigma}_i^{-1}$ and $\Sigma_i^{-1}$ decays superlinearly with $T$.

\textit{Step 3.} 
Finally, we concatenate $[\wh{G}-G]_i$ for all $i\in\calI\coloneqq\{i_1,\cdots,i_{|\calI|}\}$ and upper bound the norm of the concatenation. 
\begin{equation}\begin{split}\label{eq:markov_text5}
    {}& \norm{\left[\wh{G}-G\right]_{\calI}} = \norm{\begin{bmatrix}
        [\Delta]_{i_1} \wh{\Sigma}_{i_1}^{-1}\\
        \vdots\\
        [\Delta]_{i_{|\calI|}} \wh{\Sigma}_{i_{|\calI|}}^{-1}\\
    \end{bmatrix}}\\
    \leq {}& \norm{\begin{bmatrix}
        [\Delta]_{i_1}\Sigma_{i_1}^{-1}\\
        \vdots\\
        [\Delta]_{i_{|\calI|}} \Sigma_{i_{|\calI|}}^{-1}\\
    \end{bmatrix}} + \norm{\begin{bmatrix}
        [\Delta]_{i_1} \b{\wh{\Sigma}_{i_1}^{-1}-\Sigma_{i_1}^{-1}}\\
        \vdots\\
        [\Delta]_{i_{|\calI|}} \b{\wh{\Sigma}_{i_{|\calI|}}^{-1}-\Sigma_{i_{|\calI|}}^{-1}}\\
    \end{bmatrix}}.
\end{split}\end{equation}

In the last line of \Cref{eq:markov_text5}, the first term equals the following matrix by the definition $\Sigma_i^{-1}=1/\b{|\calK_i|(T-d+1)\sigma_u^2}I$:
\begingroup\makeatletter\def\f@size{9}\check@mathfonts\makeatother
\begin{equation*}\begin{split}
    \norm{\begin{bmatrix}
        \dfrac{1}{|\calK_{i_1}|(T-d+1)\sigma_u^2}[\Delta]_{i_1}\\
        \vdots\\
        \dfrac{1}{|\calK_{i_{|\calI|}}|(T-d+1)\sigma_u^2}[\Delta]_{i_{|\calI|}}\\
    \end{bmatrix}} \leq \dfrac{1}{s_{\min}(T-d+1)\sigma_u^2}\norm{\Delta}.
\end{split}\end{equation*}
\endgroup
The upper bound holds because $s_{\min}= \min_{i\in\calI}|\calK_i|$.

In the last line of \Cref{eq:markov_text5}, the second term is upper bounded by 
\begingroup\makeatletter\def\f@size{9}\check@mathfonts\makeatother
\begin{equation*}\begin{split}
    {}& \sqrt{|\calI|}\cdot \max_{i\in\calI}\norm{[\Delta]_i \b{\wh{\Sigma}_i^{-1}-\Sigma_i^{-1}}}\\
    \leq {}& \sqrt{r}\norm{\Delta}\max_{i\in\calI}\norm{\wh{\Sigma}_i^{-1}-\Sigma_i^{-1}} \leq \frac{2\kappa_5}{\sigma_u^4}\sqrt{\dfrac{rmdT}{s_{\min}^2(T-d+1)^4}}\norm{\Delta}.
\end{split}\end{equation*}
\endgroup
Here the last inequality is by \Cref{eq:markov_10}. 

Substitute the two bounds into \Cref{eq:markov_text5} and we get 
\begingroup\makeatletter\def\f@size{9}\check@mathfonts\makeatother
\begin{equation*}\begin{split}
    {}& \norm{\left[\wh{G}-G\right]_{\calI}}\\
    \leq {}& \dfrac{1}{s_{\min}(T-d+1)\sigma_u^2}\norm{\Delta} + \frac{2\kappa_5}{\sigma_u^4}\sqrt{\dfrac{rmdT}{s_{\min}^2(T-d+1)^4}}\norm{\Delta}.
\end{split}\end{equation*}
\endgroup
When \Cref{eq:markov_1} holds, the second term is smaller than the first term. Therefore, we conclude that
\begin{equation*}\begin{split}
    \norm{\left[\wh{G}-G\right]_{\calI}} \leq {}& \dfrac{2}{s_{\min}(T-d+1)\sigma_u^2}\kappa_4\sqrt{mds_{\max}T}\\
    \leq {}& \kappa_1 \sqrt{\frac{s_{\max}}{s_{\min}}\frac{md}{s_{\min}T}}.
\end{split}\end{equation*}
In the last line we have used $T-d+1\geq T/2$.\qed

\subsection[]{Proof of \Cref{lem:param}}
\begin{proof}
    For notational simplicity, we first define
    \begin{equation*}\begin{split}
        G \coloneqq {}& \begin{bmatrix}
            B & \cdots & A^{d}B
        \end{bmatrix}, G^- \coloneqq \begin{bmatrix}
            B & \cdots & A^{d-1}B
        \end{bmatrix},\\
        G^+ \coloneqq {}& \begin{bmatrix}
            AB & \cdots & A^{d}B
        \end{bmatrix} = AG^-.\\
    \end{split}\end{equation*}
    Because the system is controllable and $d\geq r$, $G^-$ is full-row-rank and has right pseudo-inverse, denoted as $(G^-)^\dagger$.

    We will also use the following notations where $\htG$ is the output of \Cref{alg:alloc}:
    \begin{equation*}\begin{split}
        \wh{G} \coloneqq {}& \begin{bmatrix}
            \wh{B} & \cdots & \wh{A^{d}B}
        \end{bmatrix},
        \wh{G}^- \coloneqq \begin{bmatrix}
            \wh{B} & \cdots & \wh{A^{d-1}B}
        \end{bmatrix},\\
        \wh{G}^+ \coloneqq {}& \begin{bmatrix}
            \wh{AB} & \cdots & \wh{A^{d}B}
        \end{bmatrix}.
    \end{split}\end{equation*}
    Since $\htG^--G^-$ is a submatrix of $\htG-G$, we know that
    \begin{equation*}\begin{split}
        \norm{\htG^--G^-} \leq \norm{\htG-G} \leq \frac{\sigma_{\min}(G^-)}{2},
    \end{split}\end{equation*}
    where the last inequality holds by \Cref{eq:param_2}.
    Therefore, $\sigma_r(\htG^-)$ is lower bounded as follows 
    \begin{equation}\begin{split}\label{eq:text10}
        {}& \sigma_{r}\b{\wh{G}^-} \geq \sigma_r\b{G^-} - \norm{\wh{G}^- - G^-}\\
        = {}& \sigma_{\min}\b{G^-} - \norm{\wh{G}^- - G^-} \geq \frac{\sigma_{\min}\b{G^-}}{2}. 
    \end{split}\end{equation}
    We then conclude that $\wh{G}^-$ is full-row-rank and has right pseudo-inverse, denoted by $(\wh{G}^-)^\dagger$.

    We now upper bound $\|\wh{A}-A\|$. 
    Since $\wh{A}=\wh{G}^+(\wh{G}^-)^{\dagger}$
    and $A=G^+(G^-)^\dagger$, we know that
    \begin{equation}\begin{split}\label{eq:text11}
        {}& \norm{\wh{A} - A} = \norm{\wh{G}^+\big(\wh{G}^-)^{\dagger} - G^+\b{G^-}^\dagger}\\
        \leq {}& \norm{\big(\wh{G}^+ - G^+)\big(\wh{G}^-\big)^\dagger} + \norm{G^+\b{\big(\wh{G}^-\big)^\dagger - \b{G^-}^\dagger}}\\
        \leq {}& \norm{\wh{G}^+ - G^+}\norm{\big(\wh{G}^-\big)^\dagger} + \norm{G^+}\norm{\big(\wh{G}^-\big)^\dagger - \b{G^-}^\dagger}\\
    \end{split}\end{equation}
    For the first term, we have
    \begin{equation*}\begin{split}
        {}& \norm{\wh{G}^+ - G^+} \leq \norm{\wh{G}-G},\\
        {}& \norm{\big(\wh{G}^-\big)^\dagger} = \frac{1}{\sigma_r(\htG^-)} \leq \frac{2}{\sigma_{\min}\b{G^-}}.
    \end{split}\end{equation*}
    Here the first inequality is because $\wh{G}^+-G^+$ is a submatrix of $\wh{G}-G$, and the second inequality is by \Cref{eq:text10}.
    For the second term, we have
    \begin{equation*}\begin{split}
        \norm{\big(\wh{G}^-\big)^\dagger - \big(G^-\big)^\dagger}
        \overset{(i)}{\leq} {}& 2\norm{\big(\wh{G}^-\big)^\dagger}\norm{\big(G^-\big)^\dagger}\norm{\wh{G}^- - G^-}\\
        \leq {}& 2\frac{2}{\sigma_{\min}(G^-)}\frac{1}{\sigma_{\min}(G^-)}\norm{\wh{G}^- - G^-}\\
        \leq {}& \frac{4}{\sigma_{\min}^2\b{G^-}}\norm{\wh{G}-G}.
    \end{split}\end{equation*}
    Here $(i)$ is by \cite[Theorem 4.1]{proof_3wedin1973perturbation}.
    Substituting upper bounds of the two terms back into \Cref{eq:text11} gives
    \begin{equation*}\begin{split}
        \norm{\wh{A} - A} \leq {}& \b{\frac{2}{\sigma_{\min}\big(G^-\big)} + \frac{4\norm{G}}{\sigma_{\min}^2\b{G^-}}}\norm{\wh{G}-G}\\
        \leq {}& \frac{6\norm{G}}{\sigma_{\min}^2\big(G^-\big)}\norm{\wh{G}-G}.
    \end{split}\end{equation*}
    In the last line, we have used $\sigma_{\min}\b{G^-} \leq \norm{G^-} \leq \norm{G}$.

    For $\wh{B}$, we note that $\wh{B}-B = [\wh{G}]^{1:m} - [G]^{1:m} = [\wh{G}-G]^{1:m}$ is a submatrix of $\wh{G}-G$. Therefore $\|\wh{B}-B\|\leq \|\wh{G}-G\|$.
\end{proof}

\subsection{Proof of \Cref{thm:alloc}}
\begin{proof}
    For simplicity, we let $c_1 =s_{\max}md/s_{\min}$ and $\rho_{\Delta} \coloneqq \norm{\htA-A} \leq \kappa_2\sqrt{c_1/(s_{\min}T)}$.

    \textbf{Step 1: We first prove \Cref{eq:sub2}.} To do this, we first upper bound error of the estimated observability matrix $\htO(C)$ for any non-zero measurement matrix $C$.
    \begin{equation*}\begin{split}
        \norm{\htO(C)-O(C)} \leq \sum_{i=1}^{r-1} \norm{C(A^{i}-\htA^i)} \leq \sum_{i=1}^{r-1} \norm{A^{i}-\htA^i}.
    \end{split}\end{equation*}    
    \noindent We relate error in $\htA^i$ to error in $\htA$ with the following lemma:
    \begin{lemma}\label{lem:power_pert}
        Consider any matrix $A$ satisfying Assumption \ref{assmp:stable} with constants $\psi_A,\rho_A$. Consider any perturbation $\Delta$ satisfying $\norm{\Delta} \leq \rho_{\Delta}$. The following holds for any positive integer $n$:
        \begin{equation}\begin{split}
            \norm{(A+\Delta)^n - A^n} \leq n\b{\frac{\psi_A}{\rho_A}}^2\b{\rho_A+\frac{\psi_A}{\rho_A}\rho_{\Delta}}^{n-1}\rho_{\Delta}.\qedhere
        \end{split}\end{equation}
    \end{lemma}
    \noindent We apply the above lemma and get
    \begingroup\makeatletter\def\f@size{8}\check@mathfonts\makeatother
    \begin{equation*}\begin{split}
        {}& \norm{\htO(C)-O(C)} \leq \sum_{i=1}^{r-1} i\b{\frac{\psi_A}{\rho_A}}^2\b{\rho_A+\frac{\psi_A}{\rho_A}\rho_{\Delta}}^{i-1}\rho_{\Delta}\\
        = {}& \frac{\psi_A^2}{\rho_A^2}\rho_{\Delta}\frac{1-r(\rho_A+\psi_A\rho_{\Delta}/\rho_A)^{r-1}+(r-1)(\rho_A+\psi_A\rho_{\Delta}/\rho_A)^r}{(1-(\rho_A+\psi_A\rho_{\Delta}/\rho_A))^2}\\
        \leq {}& \frac{\psi_A^2}{\rho_A^2}\rho_{\Delta}\frac{1-(r-1)(\rho_A+\psi_A\rho_{\Delta}/\rho_A)^{r-1}\b{1-\rho_A-\psi_A\rho_{\Delta}/\rho_A}}{(1-(\rho_A+\psi_A\rho_{\Delta}/\rho_A))^2}.
    \end{split}\end{equation*}
    \endgroup
    
    \noindent By \Cref{eq:sub1}, $T > 16\psi_A^2\kappa_2^2c_1/(\rho_A(1-\rho_A))^2$ and therefore $\rho_A+\psi_A\rho_{\Delta}/\rho_A < (3\rho_A+1)/4 < 1$. It then follows that,
    \begin{equation*}\begin{split}
        {}& \norm{\htO(C)-O(C)} \leq \frac{(4/3)^2\psi_A^2}{\rho_A^2(1-\rho_A)^2}\rho_{\Delta}\\
        \leq {}& \frac{(4/3)^2\psi_A^2}{\rho_A^2(1-\rho_A)^2}\sqrt{\frac{\kappa_2^2c_1}{s_{\min}T}}< \sqrt[4]{\frac{1}{s_{\min}T}}.
    \end{split}\end{equation*}
    Here the last inequality is also by \Cref{eq:sub1}.

    Now we define $r_c\coloneqq \rank\b{O(C)}$ and get
    \begin{equation*}\begin{split}
        \sigma_{r_c}\b{\wh{O}(C)} \geq {}& \sigma_{r_c}\b{O(C)} - \norm{\wh{O}(C)-O(C)}\\
        \geq {}& \sigma_O - \sqrt[4]{\frac{1}{s_{\min}T}} > \sqrt[4]{\frac{1}{s_{\min}T}}.
    \end{split}\end{equation*}
    The last inequality is by \Cref{eq:sub1}. On the other hand, for positive integer $i > r_c$, we have $\sigma_i(O(C)) = 0$ and therefore
    \begin{equation*}\begin{split}
        \sigma_i\b{\wh{O}(C)} \leq {}& \sigma_i\b{O(C)} + \norm{\wh{O}(C)-O(C)} < \sqrt[4]{\frac{1}{s_{\min}T}}.
    \end{split}\end{equation*}
    
    On the other hand, $\htr(C)$ is defined as the largest $i$ with $\sigma_i\b{\wh{O}(C)} > \sqrt[4]{1/(s_{\min}T)}$ (\Cref{alg:rank}). Therefore, for any measurement matrix $C$, $\htr(C) = r_c = \rank\b{O(C)}$.

    \textbf{Step 2. We now prove \Cref{eq:sub5}.} This part of the proof is largely based on Theorem 1 in \cite{wolsey1982analysis}, or Section ``Submodular Optimization Algorithms'' in \cite{clark2017submodularity}. It is clear \Cref{alg:2} is the greedy algorithm for the following problem
    \begin{equation*}\begin{split}
        \min_{\calS\subseteq [r]} |\calS| \quad \text{s.t.}\quad g(\calS) \geq r,
    \end{split}\end{equation*}
    where $g\b{\calS} = \rank\b{O\b{I_{\calS}}}$ with $I_{\calS} = [I_r]_{\calS}$ representing the measurement matrix with sensors on state coordinates $\calS$.
    By Step 1, we know that $\htr(I_{\calS}) = \rank\b{O\b{I_{\calS}}} = g\b{\calS}$ for any $\calS\subseteq [r]$. This is equivalent to knowing function $g(\cdot)$.

    By Theorem 7 of \cite{summers2015submodularity}, $g(\calS)$ is submodular and monotone increasing. This satisfies the condition of Theorem 1 in \cite{wolsey1982analysis}
    We then conclude by Theorem 1 in \cite{wolsey1982analysis} that the output $\htC\in\bbR^{\htn\times r}$ of \Cref{alg:2} satisfies
    \begin{equation*}\begin{split}
        \htn \leq \b{1+\log(r)} n^*, \quad \rank(O(\htC)) \geq r.
    \end{split}\end{equation*}
    The second inequality implies $\rank(O(\htC)) = r$. 
    \qed
\end{proof}

\section[]{Extension to Systems with Inaccessible State Coordinates}\label{sec:extension}
The results in previous sections, i.e., \Cref{thm:markov}, \Cref{lem:param}, and \Cref{thm:alloc}, assume that sensors can be allocated on any state coordinate. In practice, however, certain state coordinates may be inaccessible. In this section, we extend the algorithms and analysis to handle such settings. 

Let $\calJ=\{j_1, \cdots, j_{J}\}\subset[r]$ denote the set of accessible state coordinates with cardinality $J=|\calJ|$. We make the following assumption throughout this section:
\begin{assumption}\label{assmp:obs}
    Let $I_{\calJ} \coloneqq \begin{bmatrix}
        e_{j_1}\t & \cdots & e_{j_{J}}\t
    \end{bmatrix}\t = [I_r]_{\calJ}$. Suppose $(A,I_{\calJ})$ is observable.\qed
\end{assumption}
\noindent Under this assumption, there exists at least one sensor allocation that ensures observability by placing sensors only on state coordinates in $\calJ$. Our goal is to approximate the system parameters, and find a more efficient allocation that achieves observability with fewer sensors.

\subsection{SYSID with Inaccessible State Coordinates}
As in \Cref{sec:sysid}, we follow the principle that every accessible state
coordinate in $\calJ$ is measured in at least $s$ trajectories for a chosen integer $s\geq 1$. To this end, we let $f_a \coloneqq e_{j_a}\in\bbR^r$ be the one-hot vector with $1$ on the $j_a$-th coordinate for $a\in[J]$. We also extend $f$ cyclically by $f_0\coloneqq f_J$ and $f_a\coloneqq f_{a\bmod J}$ for any integer $a>J$, so that $f_{J+1}=f_1$, $f_{J+2}=f_2$, and so on. 
We then choose the number of trajectories $K = \lceil sJ/\barn\rceil$, where $\barn$ is the number of sensors, and define the measurement matrices
\begin{equation}\label{eq:C_unacc}
    C_k = \begin{bmatrix}
        f_{(k-1)\barn+1}\t & \cdots & f_{k\barn}\t
    \end{bmatrix}\t, \quad \forall k\in[K].
\end{equation}
Intuitively, $\{C_k\}_{k\in[K]}$ cyclically allocates sensors over the accessible coordinates $j_1, \ldots, j_J$ so that each accessible coordinate is measured in at least $s$ trajectories.

We then run \Cref{alg:alloc}, with line 2 replaced by the above measurement matrices (\Cref{eq:C_unacc}). As a result, by \Cref{thm:markov}, the rows of $G$ indexed by $\calJ$, denoted $[G]_{\calJ}$, can be learned up to an error of order $\kappa_1\sqrt{md/(sT)}$.

\begin{remark}[Recovery of $A,B$ via Ho-Kalman]\label{rem:ho_kalman}
Beyond the Markov parameters $[\htG]_{\calJ}$, the system matrices $A,B$ can also be recovered, up to a similarity transformation, via the Ho-Kalman algorithm \cite[Theorem 5.3]{sysid_5}. Specifically, consider the estimated Hankel matrix $\htH(I_{\calJ})$
(defined in \Cref{eq:hankel} in the following subsection), constructed from $[\htG]_{\calJ}$. With $d\geq 2r-1$ and Assumption \ref{assmp:obs}, we can recover $\htA$ and $\htB$ by taking the rank-$r$ SVD $\htH(I_{\calJ}) = U\Sigma V\t$ and setting
\begingroup\makeatletter\def\f@size{9}\check@mathfonts\makeatother
\begin{subequations}\label{alg:recover_unacc}
\begin{empheq}[box=\fbox]{align*}
    \htR \coloneqq {}& \Sigma^{1/2}V\t,\htR^- \coloneqq [\htR]^{1:rm},\htR^+ \coloneqq [\htR]^{m+1:(r+1)m}.\\
    \htA \coloneqq {}& \htR^+ \b{\htR^-}^\dagger, ~~\htB \coloneqq [\htR^-]^{1:m}.
\end{empheq}
\end{subequations}
\endgroup
\end{remark}

\subsection{Sensor Allocation with Inaccessible State Coordinates}
We now focus on learning an efficient sensor allocation with the
estimated Markov parameters $[\htG]_{\calJ}$. The
key ingredient is the following estimated Hankel matrix
\begin{align}\label{eq:hankel}
    \htH(C) \coloneqq {}& \begin{bmatrix}
        C\wh{B} & C\wh{AB} & \cdots & C\wh{A^{r}B}\\
        C\wh{AB} & C\wh{A^2B} & \cdots & C\wh{A^{r+1}B}\\
        \vdots & \vdots & \ddots & \vdots\\
        C\wh{A^{r-1}B} & C\wh{A^{r}B} & \cdots & C\wh{A^{2r-1}B}
    \end{bmatrix}.
\end{align}
For the rest of this section, we will evaluate $\htH(C)$ at measurement matrix $C$ that only measures coordinates in $\calJ$. For such $C$, each row is $e_j\t$ for some $j \in \calJ$. Consequently, each row of $C\wh{A^iB}$ is a row of $[\wh{A^iB}]_{\calJ}$, and therefore a submatrix of $[\htG]_{\calJ}$. Therefore, $\htH(C)$ can be formed by reorganizing $[\htG]_{\calJ}$. 

We also write $H(C)$ for the corresponding true Hankel matrix, defined identically as $\htH(C)$ (\Cref{eq:hankel}) with true Markov parameters $B, \cdots, A^{2r-1}B$. When $[\wh{G}]_{\calJ}$ is close to the true Markov parameters $[G]_{\calJ}$, $\htH(C)$ is also close to $H(C)$. On the other hand, because the system is controllable under Assumption \ref{assmp:ctrb}, we have $\rank(H(C))=\rank(O(C))$ via the factorization $H(C)=O(C)\begin{bmatrix}
    B & \cdots & A^{r}B
\end{bmatrix}$. Utilizing this intuition, we will design an algorithm that estimates the rank of $O(C)$ using $\htH(C)$ as a proxy for $H(C)$.

Concretely, we adapt \Cref{alg:2} by making two modifications.
\textit{First}, we restrict the search space to measurement matrices $\htC$ that only measure state coordinates in $\calJ$ by setting the coordinate set to $\calI=\calJ$.
\begin{algorithm}[ht]
    \caption{Rank Estimation Subroutine $\calE$ with Inaccessible State Coordinates}
    \label{alg:rank_unacc}
    \begin{algorithmic}[1]
    \State \textbf{Init: } receive $[\htG]_{\calJ}=[I_{\calJ}\wh{B} ~ \cdots ~ I_{\calJ}\wh{A^dB}]$ with $d\geq 2r-1$, and threshold $\sqrt[4]{1/(s_{\min}T)}$;
    \State \textbf{Input:} query measurement matrix $C$;
    \State Form Hankel matrix $\htH(C)$ (\Cref{eq:hankel}) using $[\htG]_{\calJ}$;
    \State Rank Estimation:\vspace{-0.5em}
    \begin{equation*}\begin{split}
        \htr\b{C} \gets \max \left\{i: \sigma_i\b{\wh{H}(C)} > \sqrt[4]{\frac{1}{s_{\min}T}}\right\};
    \end{split}\end{equation*}
    \State \textbf{Output:} $\htr\b{C}$.
    \end{algorithmic}
\end{algorithm}
\textit{Second}, we replace the rank estimation subroutine with \Cref{alg:rank_unacc}. The subroutine forms the estimated Hankel matrix $\htH(C)$ (line 3) and thresholds its singular values (line 4). When $\htH(C)$ is sufficiently close to the true Hankel matrix $H(C)$, the singular values of $\htH(C)$ split into two groups: $\rank(H(C))$ large values corresponding to the nonzero singular values of $H(C)$, and the remaining small values corresponding to the zero singular values of $H(C)$. The threshold $\sqrt[4]{\frac{1}{s_{\min}T}}$ is chosen to separate these two groups so that $\htr(C)=\rank(H(C))= \rank(O(C))$.
\begin{remark}
    Alternatively, one could form $\htO(C) := U\Sigma^{1/2}$ from the rank-$r$ SVD $\htH(C) = U\Sigma V\t$ and threshold its singular values, as in the balanced realization framework of \cite{proof_2tsiamis2019finite}. Doing so requires an SVD perturbation bound (e.g., \cite[Theorem 4]{proof_2tsiamis2019finite}) that introduces an additional factor of $\sqrt{r/\sigma_H}$ in the perturbation analysis, yielding a worse polynomial dependence on $r$ in \Cref{eq:sub3}.\qed
\end{remark}

We now state the guarantee for \Cref{alg:2} with coordinate set $\calI=\calJ$ and rank estimation subroutine \Cref{alg:rank_unacc}. 
\begin{theorem}\label{thm:rank2}
    Consider system $(A,B,\sigma_u^2,\sigma_w^2)$ satisfying Assumptions \ref{assmp:ctrb} and \ref{assmp:obs}. Consider Markov parameter estimate $\htG$ satisfying $\norm{[\htG-G]_{\calJ}} \leq \kappa_3\sqrt{s_{\max}md/(s_{\min}^2T)}$ for $G=\begin{bmatrix}
        B & \cdots & A^{d}B
    \end{bmatrix}$ with $d\geq 2r-1$ and some constant $\kappa_3\geq 1$. Define $\sigma_H \coloneqq \min_{C\in\calC_{\calJ}}\sigma_{\min}(H(C))$ where $\calC_{\calJ}$ denotes the set of all non-zero measurement matrices that only measure state coordinates in $\calJ$. If 
    \begingroup\makeatletter\def\f@size{9}\check@mathfonts\makeatother
    \begin{equation}\begin{split}\label{eq:sub3}
        T >\max \left\{\frac{\kappa_3^4r^2m^2d^2}{s_{\min}}\frac{s_{\max}^2}{s_{\min}^2},\frac{16}{\sigma_H^4s_{\min}}\right\},
    \end{split}\end{equation}
    \endgroup
    then $\htr(C)$ from Algorithm \ref{alg:rank_unacc} satisfies the following for any measurement matrix $C$ that only measures coordinate in $\calJ$,
    \begingroup\makeatletter\def\f@size{9}\check@mathfonts\makeatother
    \begin{equation}\begin{split}\label{eq:sub4}
        \htr(C) = {}& \rank(O(C)) = \rank\b{\begin{bmatrix}
            C\t & \cdots & (CA^{r-1})\t
    \end{bmatrix}\t}.
    \end{split}\end{equation}
    \endgroup
    Furthermore, the output $\htC\in\bbR^{\htn\times r}$ of \Cref{alg:2}, with rank-estimation subroutine \Cref{alg:rank_unacc}, satisfies
    \begin{equation}\begin{split}\label{eq:sub6}
        (A, \htC) \text{ is observable}, \quad \htn \leq \b{1+\log(r)}\cdot \tiln^*.
    \end{split}\end{equation}
    Here $\tiln^*$ is the minimum number of sensors needed to make the system observable when sensors can only be placed on state coordinates in $\calJ$.\qed
\end{theorem}

\section{Numerical Experiments}\label{sec:sim}
We evaluate our main algorithms (\Cref{alg:alloc,alg:2}) on the following two models.
\textbf{Model 1: block-cyclic system, all coordinates accessible.} We let $A=0.9I_5\otimes A_{\text{blk}}$ and $B=I_{20}$, where $A_{\text{blk}} = [e_2\ e_3\ e_4\ e_1]\in\bbR^{4\times 4}$ is the cyclic-shift permutation.
The block-cyclic structure of $A$ implies observability requires at least $5$ sensors (one per block).
We run \Cref{alg:alloc} with $\barn=5$ sensors, $K\in\{4,8,16\}$ trajectories, and varying lengths $T$. The estimation error in \Cref{fig:results}(a) decays as $K$ and $T$ increase. We then run \Cref{alg:2} with subroutine \Cref{alg:rank} using Markov parameters learned at $K=16$, $T=20000$. The learned $\htC$ (\Cref{fig:results}(b)) achieves observability with an optimal number of $\htn=5$ sensors.

\textbf{Model 2: thermal dynamics with inaccessible coordinates.} We consider a $20$-zone HVAC system \cite{li2021distributed} with $20$-dimensional states (zone temperature) and inputs (zone air-flow rate).
The zones are organized into $5$ separate two-by-two blocks; within each block, only neighboring zones are considered connected. 
Let $x\in\bbR^{20}$ and $u\in\bbR^{20}$ denote the concatenation of temperatures (states) and air flow rates (inputs) of the $20$ zones. Dynamics of zone $i$ is described by 
    \begingroup\makeatletter\def\f@size{9}\check@mathfonts\makeatother
\begin{equation*}\begin{split}
    [x_{t+1}]_i ={}&  [x_t]_i + \underbrace{\frac{\Delta}{v_i\xi_i}\b{\theta - [x_t]_i}}_{\text{heat transfer with env.}} + \underbrace{\frac{\Delta}{v_i}[u_t]_i}_{\text{air flow input}}\\
    {}& + \underbrace{\sum_{j\in[20]} \frac{\Delta}{v_i\xi_{ij}}\b{[x_t]_j - [x_t]_i}}_{\text{heat transfer among zones}} + \underbrace{\frac{\sqrt{\Delta}}{v_i}[w_t]_i}_{\text{noise}}.
\end{split}\end{equation*}
\endgroup
Here we choose time resolution $\Delta=35s$, environment temperature $\theta=0^\circ\text{C}$, inter-zone thermal resistance $\xi_{ij}=1^\circ\text{C/kW}$ for connected zones (else $\infty$), thermal resistance $\xi_i=1^\circ\text{C/kW}$, thermal capacity $v_i=100\text{kJ}/^\circ\text{C}$, and variances $\sigma_w^2 = \sigma_{\eta}^2 = 1$, and $\sigma_{u}^2 = 10$. 

Only the first three coordinates in each block are accessible, giving $\calJ = [20] \setminus\{4,8,12,16,20\}$. Observability requires at least $10$ sensors, two in each block. We run \Cref{alg:alloc} with $\barn=5$ sensors. The estimation error $\|[\wh{G}-G]_{\calJ}\|$ in \Cref{fig:results}(c) decays rapidly with $K$ and $T$. We then run \Cref{alg:2} with coordinate set $\calJ$ and subroutine \Cref{alg:rank_unacc} using Markov parameters learned at $K=12$, $T=20000$. The learned $\htC$ (\Cref{fig:results}(d)) achieves observability with an optimal number of $\htn=10$ sensors.

\begin{figure}[!t]\centering
\setlength{\tabcolsep}{2pt}
\newlength{\plotH}\settoheight{\plotH}{\includegraphics[width=0.45\linewidth]{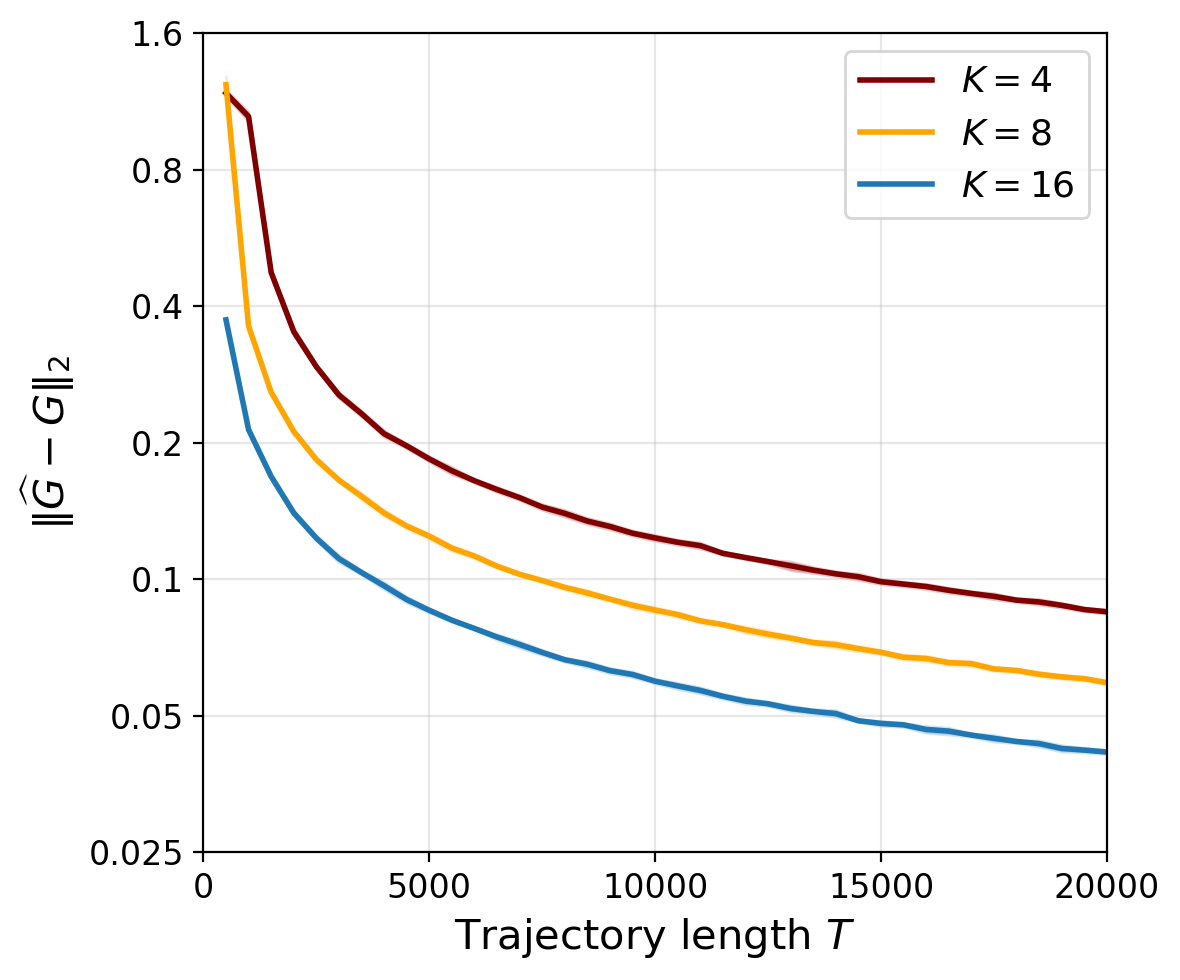}}
\newlength{\allocH}\settoheight{\allocH}{\includegraphics[width=0.4\linewidth]{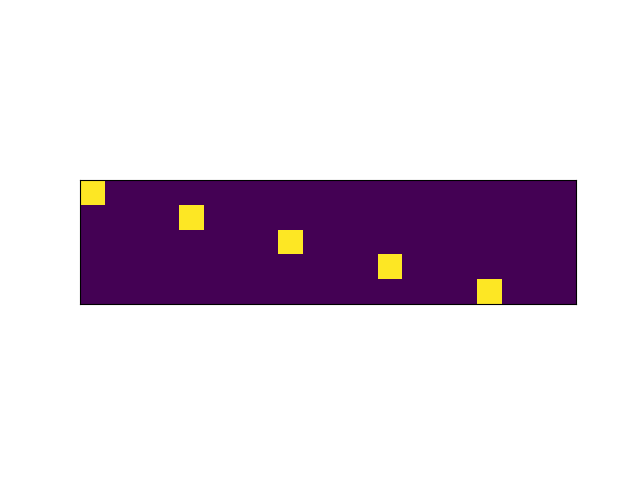}}
\begin{tabular}{c@{~}c@{\hspace{6pt}}c@{~}c}
\raisebox{\dimexpr \plotH-20pt\relax}{(a)} & \includegraphics[width=0.4\linewidth]{figs/Gerr_logy_model1.png} &
\raisebox{\dimexpr \plotH-20pt\relax}{(b)} & \includegraphics[width=0.4\linewidth]{figs/C.png} \\[2pt]
\raisebox{\dimexpr \plotH-20pt\relax}{(c)} & \includegraphics[width=0.4\linewidth]{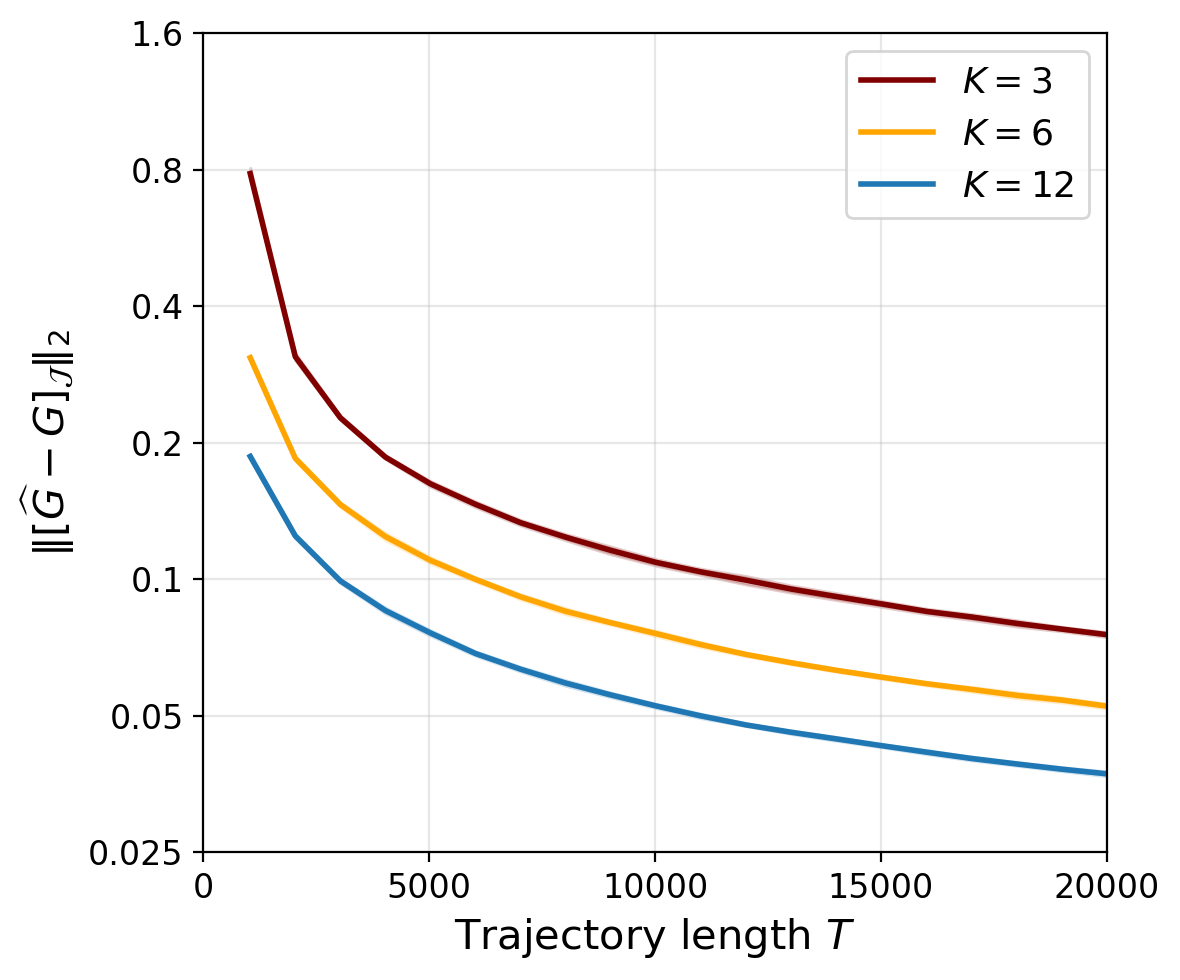} &
\raisebox{\dimexpr \plotH-20pt\relax}{(d)} & \includegraphics[width=0.4\linewidth]{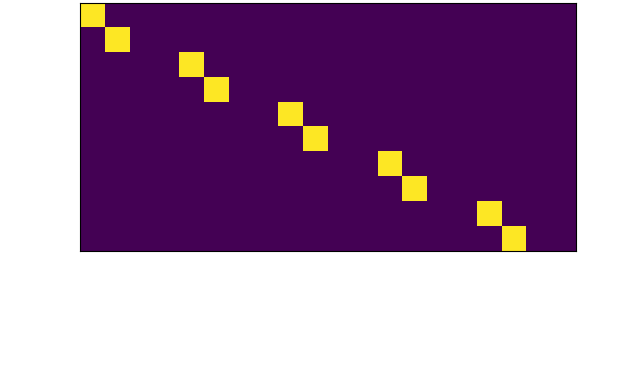}
\end{tabular}
\caption{Simulation results for Model~1 (top row, block-cyclic system) and Model~2 (bottom row, thermal dynamics). Panels (a,c) plot Markov parameter estimation errors against trajectory length $T$ on a logarithmic $y$-axis. Panels (b) and (d) visualize the learned sensor allocation matrix when $K=16$ or $K=12$, respectively. The bright squares mark active (sensor, state-coordinate) pairs.}
\label{fig:results}
\vspace{-15pt}
\end{figure}

\vspace{-0.5em}\section{Conclusions}
In this paper, we address the problem of learning efficient sensor allocations that guarantee observability with a small number of sensors for unknown high-dimensional linear systems.
Our two-stage solution consists of a novel SYSID algorithm that integrates information from multiple data trajectories, each potentially observing a different subset of the state coordinates, and a sensor allocation algorithm that operates on approximated system parameters. When sensors can be allocated on any state coordinate, the SYSID algorithm is guaranteed to learn the system parameters up to an error of $\tilde{\calO}(\sqrt{md/(s_{\min}T)})$, where $m$ is the input dimension, $d\geq r$ is a hyperparameter with $r$ being the state dimension, $s_{\min}$ is the least number of trajectories used to measure one state coordinate, and $T$ is the trajectory length. Furthermore, the sensor allocation algorithm is guaranteed to output an allocation with at most $(1+\log(r))n^*$ sensors, where $n^*$ is the minimum number of sensors required to ensure observability. We also extend our algorithms to scenarios where certain state coordinates are inaccessible. The efficiency of our algorithms is validated through numerical experiments.

\appendices

\section{Proofs for \Cref{proof:thm1}}\label{proof:thm}
\begin{proof}[Proof of \Cref{lem:cross_all}]
    Recall $\Delta$ (\Cref{eq:markov_text}) is defined as
    \begin{align}
        \Delta = {}& \sum_{k=1, t=d}^{K,T} C_k\t C_k \b{\Delta_{k,t+1}^1 + \Delta_{k,t+1}^2}U_{k,t}(d)\t,
    \end{align}
    where $\Delta_{k,t}^1$ and $\Delta_{k,t}^2$ (in \Cref{eq:markov_text7}) are defined as
    \begin{equation}\begin{split}
        \Delta_{k,t+1}^1 = {}& \sum_{\tau=d+1}^{t} A^{\tau}\b{Bu_{k,t-\tau} + w_{k,t-\tau}},\\
        \Delta_{k,t+1}^2 = {}& \eta_{k,t+1} + \sum_{\tau=0}^d A^{\tau}w_{k,t-\tau}.
    \end{split}\end{equation}

    We now upper bound the first term in the definition of $\Delta$. For simplicity, we first let $\Sigma_w \coloneqq \b{\sigma_u^2 BB\t + \sigma_w^2 I_r}^{1/2}$ and define auxiliary variable $\xi_{k,t} \coloneqq \Sigma_w^{-1}\b{Bu_{k,t} + w_{k,t}}\overset{\text{i.i.d.}}{\sim}\calN(0,I_r)$. Then rearranging the term gives
    \begingroup\makeatletter\def\f@size{8}\check@mathfonts\makeatother
    \begin{align*}
        {}&  J_1 \coloneqq \sum_{k=1,t=d}^{K,T} C_k\t C_k \Big(\sum_{\tau=d+1}^{t} A^{\tau}\b{Bu_{k,t-\tau} + w_{k,t-\tau}}\Big) U_{k,t}(d)\t\\
        = {}& \sum_{k=1, t=d+1}^{K,T} C_k\t C_k \begin{bmatrix}
            A^{d+1}\Sigma_w & \cdots & A^{t}\Sigma_w
        \end{bmatrix}\begin{bmatrix}
            \xi_{k,t-d-1}\\
            \xi_{k,t-d-2}\\
            \vdots\\
            \xi_{k,0}
        \end{bmatrix} U_{k,t}(d)\t.
    \end{align*}
    \endgroup
    We apply \Cref{lem:cross1_conc} with $t_0 = d+1$, $\Sigma = \Sigma_w$, and $\psi_\Sigma^2=\psi_B^2\sigma_u^2 + \sigma_w^2$. Since $\xi_{k,t}$ is a linear function of $(u_{k,t}, w_{k,t})$, it is independent of $u_{k,t+1:T}$ and of all $u_{k',t'}$ with $k'\neq k$. This fact, together with \Cref{eq:markov_1}, satisfies the condition of \Cref{lem:cross1_conc}. Thus, with probability at least $1-\delta/6$,
    \begingroup\makeatletter\def\f@size{8}\check@mathfonts\makeatother
    \begin{equation*}\begin{split}
        \norm{J_1} \leq {}& \sqrt{\frac{24\sigma_u^2\psi_A^2(\psi_B^2\sigma_u^2+\sigma_w^2)}{(1-\rho_A)^2}\rho_A^{2d}}\sqrt{mds_{\max}(T-d)\log\b{\frac{120}{\delta}}}.
    \end{split}\end{equation*}
    \endgroup

    Now consider the second term in $\Delta$. Rearraging gives
    \begingroup\makeatletter\def\f@size{8}\check@mathfonts\makeatother
    \begin{align*}
        {}& \sum_{k=1,t=d}^{K,T} C_k\t C_k \Delta_{k,t+1}^2U_{k,t}(d)\t\\
        = {}& \sum_{k=1,t=d}^{K,T} C_k\t C_k \b{\sum_{\tau=0}^d A^{\tau}w_{k,t-\tau}} U_{k,t}(d)\t\\
        {}& + \sum_{k=1,t=d}^{K,T} C_k\t C_k \eta_{k,t+1} U_{k,t}(d)\t\\
        = {}& \underbrace{\sum_{k=1,t=d}^{K,T} C_k\t C_k \begin{bmatrix}
            I & A & \cdots & A^d\\
        \end{bmatrix} \begin{bmatrix}
            w_{k,t}\\
            w_{k,t-1}\\
            \vdots\\
            w_{k,t-d}
        \end{bmatrix} U_{k,t}(d)\t}_{J_2}\\
        {}&  + \underbrace{\sum_{k=1,t=d}^{K,T} C_k\t C_k \eta_{k,t+1} U_{k,t}(d)\t}_{J_3}
    \end{align*}
    \endgroup
    To upper bound $J_2$ and $J_3$, we resort to \Cref{lem:cross2_conc}.
    We apply \Cref{lem:cross2_conc} with $t_0 = d, \quad \xi_{k,t} =w_{k,t}$ and get the following with probability at least $1-\delta/6$
    \begin{equation*}\begin{split}
        \norm{J_2} \leq \sqrt{\frac{24\sigma_u^2\sigma_w^2\psi_A^2}{\rho_A^2(1-\rho_A)^2} mds_{\max}(T-d+1)\log{\b{\frac{120}{\delta}}}}.
    \end{split}\end{equation*}
    For $J_3$, we apply \Cref{lem:cross2_conc} with $t_0 = 0, \quad \xi_{k,t} = \eta_{k,t+1}$.
    Therefore, with probability at least $1-\delta/6$, 
    \begin{equation*}\begin{split}
        \norm{J_3} \leq \sqrt{\frac{24\sigma_u^2\sigma_{\eta}^2\psi_A^2}{\rho_A^2(1-\rho_A)^2} mds_{\max}(T-d+1)\log\b{\frac{120}{\delta}}}.
    \end{split}\end{equation*}

    Finally, with a union bound, we have the following with probability at least $1-\delta/2$
    \begin{equation*}\begin{split}
        \norm{\Delta} = {}& \norm{J_1+J_2+J_3}\leq 3\sqrt{\frac{24\sigma_u^2(\sigma_u^2+\sigma_w^2+\sigma_{\eta}^2)\psi_A^2\psi_B^2}{\rho_A^2(1-\rho_A)^2}}\\
        {}& \cdot \sqrt{mds_{\max}(T-d+1)\log\b{\frac{120}{\delta}}}.
    \end{split}\end{equation*}
\end{proof}

\begin{proof}[Proof of \Cref{lem:power_pert}]
    We first prove the following equation.
    \begin{equation}\begin{split}\label{eq:power_pert}
        \norm{\b{A+\Delta}^n} \leq \frac{\psi_A}{\rho_A}\b{\rho_A+\frac{\psi_A}{\rho_A}\rho_\Delta}^n.
    \end{split}\end{equation}
    In the binomial expansion of $(A+\Delta)^n$, consider any term with $n-i$ occurrences of $A$ and $i$ occurrences of $\Delta$. It must take the form of $A^{\alpha_0}\Delta^{\beta_0}A^{\alpha_1}\cdots A^{\alpha_l}\Delta^{\beta_l}A^{\alpha_{l+1}}$. Here $l$, $\alpha_0$, $\alpha_{l+1}$ are non-negative integers, while other coeffcients are strictly positive integers. Moreover, we have $l+1\leq i$, $\sum_{j=0}^{l}\beta_j = i$ and $\sum_{j=0}^{l+1}\alpha_j = n-i$. Taking its norm gives that 
    \begin{equation*}\begin{split}
        {}& \norm{A^{\alpha_0}\cdots A^{\alpha_{l+1}}} \leq \norm{A^{\alpha_{l+1}}}\prod_{j=0}^{l}\norm{A^{\alpha_j}}\norm{\Delta^{\beta_j}}\\
        \leq {}& \b{\frac{\psi_A}{\rho_A}}^{l+2}\rho_A^{\sum_{j=0}^{l+1}\alpha_j}\rho_{\Delta}^{\sum_{j=0}^{l}\beta_j} \leq \b{\frac{\psi_A}{\rho_A}}^{i+1}\rho_A^{n-i}\rho_{\Delta}^{i}.
    \end{split}\end{equation*}
    Notice that this upper bound holds for all terms in $(A+\Delta)^n$ with $i$ occurrences of $\Delta$. Then, 
    \begin{equation*}\begin{split}
        {}& \norm{\b{A+\Delta}^n} \leq \sum_{i=0}^{n} {n\choose i} \b{\frac{\psi_A}{\rho_A}}^{i+1}\rho_A^{n-i}\rho_{\Delta}^i\\
        = {}& \frac{\psi_A}{\rho_A}\sum_{i=0}^n {n\choose i} \rho_A^{n-i}\b{\frac{\psi_A}{\rho_A}\rho_{\Delta}}^i= \frac{\psi_A}{\rho_A}\b{\rho_A+\frac{\psi_A}{\rho_A}\rho_{\Delta}}^n.
    \end{split}\end{equation*} 

    Note that for any matrices $M_1,M_2$: $M_1^n - M_2^n = \sum_{i=0}^{n-1} M_1^{n-1-i}(M_1-M_2)M_2^i$.
    Applying the equation with $M_1=A+\Delta$ and $M_2=A$ gives
    \begingroup\makeatletter\def\f@size{9}\check@mathfonts\makeatother
    \begin{equation*}\begin{split}
        \norm{(A+\Delta)^n - A^n} \leq {}& \sum_{i=0}^{n-1} \norm{(A+\Delta)^{n-1-i}(\Delta)A^i}\\
        \leq {}& \sum_{i=0}^{n-1} \rho_{\Delta}\b{\frac{\psi_A}{\rho_A}}^2\rho_A^i\b{\rho_A+\frac{\psi_A}{\rho_A}\rho_{\Delta}}^{n-1-i}\\
        \leq {}& n\b{\frac{\psi_A}{\rho_A}}^2\b{\rho_A+\frac{\psi_A}{\rho_A}\rho_{\Delta}}^{n-1}\rho_{\Delta}.
    \end{split}\end{equation*}
    \endgroup
\end{proof}

\vspace{-2em}
\subsection{Supporting Details}
\begin{lemma}[Concentration of Covariance]\label{lem:input_conc}
    Consider i.i.d. random vectors $\{u_{k,t}\in\bbR^{m}\}_{k=1,t=0}^{K,T}$ with $u_{k,t}\sim\calN(0,\sigma_u^2I_m)$ for positive integers $m,K,T$. Consider positive integers $t_0, d$ with $t_0\geq d$, and consider any $\delta\in(0,1)$. 
    If
    \begin{equation}\begin{split}
        KT \geq 2Kt_0 + 2cmd\cdot \b{\log^4(2mdKT)+\log\b{\frac{1}{\delta}}},
    \end{split}\end{equation}
    where $c$ is the absolute constant in \cite[Theorem 4.1 with $L=1$]{krahmer2014suprema},
    then with probability at least $1-\delta$ for any $\delta\in(0,1)$,
    \begin{equation}\begin{split}\label{eq:input_conc_1}
        {}& \norm{\frac{1}{K(T-t_0+1)}\sum_{k=1,t=t_0}^{K,T}U_{k,t}(d)U_{k,t}(d)\t - \sigma_u^2 I}\\
        \leq {}& \sigma_u^2\sqrt{\frac{4md}{KT}}\cdot \sqrt{c\b{\log^4(2mdKT)+\log\b{\frac{1}{\delta}}}}.
    \end{split}\end{equation}
    Here $U_{k,t}(d) = \begin{bmatrix}
        u_{k,t}\t & \cdots & u_{k,t-d}\t
    \end{bmatrix}\t$. Moreover, if $KT \geq 2Kt_0 + 16cmd\cdot (\log^4(2mdKT)+\log(1/\delta))$,
    \begingroup\makeatletter\def\f@size{9}\check@mathfonts\makeatother    \begin{equation}\begin{split}\label{eq:input_conc_2}
        \frac{\sigma_u^2}{2}I \preceq \frac{1}{K(T-t_0+1)}\sum_{k=1,t=t_0}^{K,T}U_{k,t}(d)U_{k,t}(d)\t \preceq \frac{3\sigma_u^2}{2}I.
    \end{split}\end{equation}
    
    \endgroup
\end{lemma}

\begin{proof}
    \textbf{Step 1: We first prove \Cref{eq:input_conc_1}.}
    The input covariance can be rewritten as
    \begin{equation*}\begin{split}
        {}& \sum_{k=1,t=t_0}^{K,T}U_{k,t}(d)U_{k,t}(d)\t = \Phi\t\Phi.
    \end{split}\end{equation*}
    
    \vspace{-\parskip}\noindent Here we have defined $\Phi\t \coloneqq \left[
    U_{1,t_0}(d)\cdots U_{1,T}(d) ~U_{2,t_0}(d) \right.$ $ \left. \cdots U_{2,T}(d) \cdots U_{K,t_0}(d) \cdots U_{K,T}(d)\right]\in\bbR^{m(d+1)\times K(T-t_0+1)}$. 
    Observe $\Phi$ is a $K(T-t_0+1)\times m(d+1)$ submatrix of the circulant matrix with first row $\left[U_{1,t_0}(d)\t \cdots U_{1,T}(d)\t U_{2,t_0}(d)\t\cdots \right.$
    $\left.U_{2,T}(d)\t \cdots U_{K,t_0}(d)\t  \cdots U_{K,T}(d)\t\right]\in\bbR^{m(d+1)K(T-t_0+1)}$. Now by \Cref{lem:circ}, if 
    \begingroup\makeatletter\def\f@size{9}\check@mathfonts\makeatother
    \begin{equation*}\begin{split}
        {}&  K(T-t_0+1) \geq cm(d+1)\\
        {}& \cdot\b{\log^2(m(d+1))\log^2(m(d+1)K(T-t_0+1))+\log\b{\frac{1}{\delta}}},
    \end{split}\end{equation*}
    \endgroup
    with probability at least $1-\delta$,
    \begingroup\makeatletter\def\f@size{9}\check@mathfonts\makeatother
    \begin{equation*}\begin{split}
        {}& \norm{\frac{1}{K(T-t_0+1)}\Phi\t\Phi - \sigma_u^2 I} \leq \sigma_u^2\sqrt{c\frac{m(d+1)}{K(T-t_0+1)}}\\
        {}& \sqrt{\b{\log^2(m(d+1))\log^2(m(d+1)K(T-t_0+1))+\log\b{\frac{1}{\delta}}}}.
    \end{split}\end{equation*}
    \endgroup
    for constant $c$ in \cite[Theorem 4.1 with $L=1$]{krahmer2014suprema}.
    As a simplified result, if 
    \begin{equation}\begin{split}
        KT \geq 2Kt_0 + 2cmd\cdot \b{\log^4(2mdKT)+\log\b{\frac{1}{\delta}}}
    \end{split}\end{equation}
    then the following holds with probability at least $1-\delta$
    \begin{equation}\begin{split}
        {}& \norm{\frac{1}{K(T-t_0+1)}\sum_{k=1,t=t_0}^{K,T}U_{k,t}(d)U_{k,t}(d)\t - \sigma_u^2 I} \\
        \leq {}& \sigma_u^2\sqrt{\frac{4md}{KT}}\cdot \sqrt{c\b{\log^4(2mdKT)+\log\b{\frac{1}{\delta}}}}.
    \end{split}\end{equation}
    \textbf{Step 2: We now prove \Cref{eq:input_conc_2}.} For $KT \geq 2Kt_0 + 16cmd\cdot \b{\log^4(2mdKT)+\log\b{\frac{1}{\delta}}}$, the following holds with probability at least $1-\delta$
    \begin{equation*}\begin{split}
        {}& \norm{\frac{1}{K(T-t_0+1)}\sum_{k=1,t=t_0}^{K,T}U_{k,t}(d)U_{k,t}(d)\t - \sigma_u^2I} \leq \frac{\sigma_u^2}{2}.
    \end{split}\end{equation*}
    Therefore, it naturally follows that
    \begin{equation*}\begin{split}
        {}& \frac{\sigma_u^2}{2}I \preceq \frac{1}{K(T-t_0+1)}\sum_{k=1,t=t_0}^{K,T}U_{k,t}(d)U_{k,t}(d)\t \preceq \frac{3\sigma_u^2}{2} I.
    \end{split}\end{equation*}
\end{proof}

\begin{lemma}[Concentration of Crossterms with Growing Dimensions]\label{lem:cross1_conc}
    Consider system $(A,B,\sigma_u^2,\sigma_w^2)$ and measurement matrices $\{C_k\}_{k\in[K]}$ as in \Cref{thm:markov}. Consider trajectory length $T$ and estimation rank $d$ as in \Cref{alg:alloc}. 
    Consider positive integer $t_0 \in [d,T]$, matrix $\Sigma\in\bbR^{r\times r}$ with $\norm{\Sigma}\leq \psi_{\Sigma}$ for some $\psi_{\Sigma}\geq 1$, and $\{\xi_{k,t}\in\bbR^r\}_{k=1,t=0}^{K,T-t_0}$ with $\xi_{k,t}\overset{\text{i.i.d.}}{\sim}\calN(0,I_r)$.
    Consider any $\delta\in(0,1)$. If $\xi_{k,t}$ is independent of $u_{k,t+t_0-d:T}$ and $u_{k',t_0-d:T}$ ($\forall k'\neq k$) and if
    \begingroup\makeatletter\def\f@size{9}\check@mathfonts\makeatother
    \begin{equation}\begin{split}\label{eq:cross_conc_4}
        T \geq 2t_0 + 16c\cdot rmd\cdot \b{\log^4(2mdT) + \log\b{\frac{10K}{\delta}}},
    \end{split}\end{equation}
    \endgroup
    where $c$ is the absolute constant in \cite[Theorem 4.1 with $L=1$]{krahmer2014suprema}, then with probability at least $1-\delta$,
    \begingroup\makeatletter\def\f@size{9}\check@mathfonts\makeatother
    \begin{equation*}\begin{split}
        {}& \norm{\sum_{k=1,t=t_0}^{K,T} C_k\t C_k \underline{G}_{t-t_0} \Xi_{k,t-t_0}(t-t_0) U_{k,t}(d)\t}\\
        \leq {}& \sqrt{\frac{24\sigma_u^2\psi_A^2\psi_{\Sigma}^2}{(1-\rho_A)^2} }\sqrt{\rho_A^{2d}\max\{md,r\}s_{\max}(T-t_0+1) \log\b{\frac{20}{\delta}}}.
    \end{split}\end{equation*}
    \endgroup
    \begingroup\makeatletter\def\f@size{9}\check@mathfonts\makeatother
    Here $\underline{G}_{t-t_0} = \begin{bmatrix}
        A^{d+1}\Sigma & \cdots & A^{d+1+t-t_0}\Sigma
    \end{bmatrix}$, $\Xi_{k,t-t_0}(t-t_0) = \begin{bmatrix}
        \xi_{k,t-t_0}\t & \cdots & \xi_{k,0}\t
    \end{bmatrix}\t$, $U_{k,t}(d) = \begin{bmatrix}
        u_{k,t}\t & \cdots & u_{k,t-d}\t
    \end{bmatrix}\t$.
    \endgroup
\end{lemma}

\begin{proof}
    From \Cref{lem:gen_norm2Net}, there exists a unit vector $v\in\bbS^{r-1}$ 
    \begingroup\makeatletter\def\f@size{8}\check@mathfonts\makeatother
    \begin{equation}\begin{split}\label{eq:cross_conc_9}
        {}& \bbP\b{\norm{\sum_{k=1,t=t_0}^{K,T} C_k\t C_k \underline{G}_{t-t_0} \Xi_{k,t-t_0}(t-t_0) U_{k,t}(d)\t} > z}\\
        \leq {}& 5^{r} \bbP\b{\norm{\sum_{k=1,t=t_0}^{K,T} v\t C_k\t C_k \underline{G}_{t-t_0} \Xi_{k,t-t_0}(t-t_0) U_{k,t}(d)\t} > \frac{1}{2} z}.
    \end{split}\end{equation}
    \endgroup

    \textbf{Step 1: We first upper bound $\|\bm{\sum_{k=1,t=t_0}^{K,T} v\t C_k\t C_k \underline{G}_{t-t_0} \Xi_{k,t-t_0}(t-t_0) U_{k,t}(d)\t}\|$ using the well-established self-normalizing bound \cite[Theorem 1]{abbasi2011improved}. } For the above unit vector $v\in\bbS^{r-1}$, we define $g_{k,\tau}\t \coloneqq v\t C_k\t C_k A^{d+1+\tau}\Sigma$ for $\tau\in[0,T-t_0]$, and get $v\t C_k\t C_k \underline{G}_{t-t_0}
        = \begin{bmatrix}
            g_{k,0}\t & \cdots & g_{k,t-t_0}\t
        \end{bmatrix}.$
    We then rewrite the objective term as follows
    \begingroup\makeatletter\def\f@size{9}\check@mathfonts\makeatother
    \begin{equation}\begin{split}\label{eq:cross_conc_13}
        {}& \sum_{k=1,t=t_0}^{K,T} v\t C_k\t C_k \underline{G}_{t-t_0} \Xi_{k,t-t_0}(t-t_0) U_{k,t}(d)\t\\
        = {}& \sum_{k=1,t=t_0}^{K,T} \b{\begin{bmatrix}
            g_{k,0}\t & \cdots & g_{k,t-t_0}\t
        \end{bmatrix}\begin{bmatrix}
            \xi_{k,t-t_0}\\
            \vdots\\
            \xi_{k,0}
        \end{bmatrix}} U_{k,t}(d)\t.
    \end{split}\end{equation}
    \endgroup
    Taking the transpose on both sides gives
    \begin{align*}
        {}& \b{\sum_{k=1,t=t_0}^{K,T} v\t C_k\t C_k \underline{G}_{t-t_0} \Xi_{k,t-t_0}(t-t_0) U_{k,t}(d)\t}\t\\
        = {}& \sum_{k=1,t=t_0}^{K,T} U_{k,t}(d)\b{\sum_{\tau=0}^{t-t_0}g_{k,\tau}\t\xi_{k,t-t_0-\tau}} \\
        = {}& \sum_{k=1}^K \begin{bmatrix}
            U_{k,t_0}(d) & U_{k,t_0+1}(d) & \cdots & U_{k,T}(d)
        \end{bmatrix}\\
        {}& \hspace{-1em}\begin{bmatrix}
            0 & \cdots & 0 & g_{k,0}\t\\
            \vdots & \iddots & g_{k,0}\t & g_{k,1}\t\\
            0 & \iddots & \iddots & \vdots\\
            g_{k,0}\t & g_{k,1}\t & \cdots & g_{k,T-t_0}\t\\
        \end{bmatrix} \begin{bmatrix}
            \xi_{k,T-t_0}\\
            \xi_{k,T-t_0-1}\\
            \vdots\\
            \xi_{k,0}
        \end{bmatrix}\\
        = {}& \sum_{k=1}^K\underbrace{ \begin{bmatrix}
            U_{k,t_0}(d) & \cdots & U_{k,T}(d)
        \end{bmatrix}}_{\calU_k}\bbI\b{v\t C_k\t C_k} \\
        {}& \hspace{-1em}\underbrace{\begin{bmatrix}
            0 & \cdots & 0 & A^{d+1}\Sigma\\
            \vdots & \iddots & A^{d+1}\Sigma & A^{d+2}\Sigma\\
            0 & \iddots & \iddots & \vdots\\
            A^{d+1}\Sigma & A^{d+2}\Sigma & \cdots & A^{d+T-t_0+1}\Sigma\\
        \end{bmatrix}}_{\calT}\begin{bmatrix}
            \xi_{k,T-t_0}\\
            \xi_{k,T-t_0-1}\\
            \vdots\\
            \xi_{k,0}
        \end{bmatrix}.
    \end{align*}
    In the last line, we have used that $g_{k,\tau}\t = v\t C_k\t C_k A^{d+1+\tau}\Sigma$ and that $\bbI\b{M} = \diag(\underbrace{M,M,\cdots,M}_{T-t_0+1})$.
    We further define $\calT_k = \bbI\b{v\t C_k\t C_k}\calT$ and rewrite the objective as
    \begin{equation*}\begin{split}
        {}& \b{\sum_{k=1,t=t_0}^{K,T} v\t C_k\t C_k \underline{G}_{t-t_0} \Xi_{k,t-t_0}(t-t_0) U_{k,t}(d)\t}\t\\
        = {}& \sum_{k=1}^K \calU_k\calT_k \Xi_{k,T-t_0}\b{T-t_0}\\
        = {}& \underbrace{\begin{bmatrix}
            \calU_1\calT_1 & \cdots & \calU_K\calT_K
        \end{bmatrix}}_{\calZ} \underbrace{\begin{bmatrix}
            \Xi_{1,T-t_0}\b{T-t_0}\\
            \vdots\\
            \Xi_{K,T-t_0}\b{T-t_0}\\
        \end{bmatrix}}_{\Xi}.
    \end{split}\end{equation*}

    We highlight that the $i$-th coordinate of $\Xi$ is independent of the first $i$ columns of $\calZ$. To see this, suppose the $i$-th coordinate of $\Xi$ is a coordinate of vector $\xi_{k,t}$ for some $k,t$. Then the $i$-th column of $\calZ$ must be composed of $U_{k,t+t_0}(d), \cdots, U_{k,T}(d)$. This is by the definition of the $(t+1)$-th column of $\calU_k\calT_k$ counted from the right. Since $\xi_{k,t}$ is independent of $u_{k,t+t_0-d:T}$, we first conclude that \textit{$\xi_{k,t}$ is independent of the $i$-th column of $\calZ$}. Now observe that all columns $i'$ of $\calZ$ with $i'<i$ only depend on inputs $u_{k',t'}$ with either $k'\neq k$ or $k'=k \cap t'>t+t_0-d$. These inputs are all independent of $\xi_{k,t}$. Therefore, \textit{$\xi_{k,t}$ is independent of the first $i$ columns of $\calZ$}.
    
    Now we resort to the self-normalizing bound \cite[Theorem 1]{abbasi2011improved} to bound $\calZ\Xi$. To apply the theorem, we construct a filtration $\{\calF_i\}_{i=0}^{r(T-t_0+1)K}$ such that 1) the $i$-th column of $\calZ$ is $\calF_{i-1}$-measurable, 2) the $i$-th coordinate of $\Xi$ is $\calF_{i}$-measurable and 3) the $i$-th coordinate of $\Xi$ is subGaussian conditioned on $\calF_{i-1}$. To achieve this, we let $\calF_i$ be the sigma algebra generated by the inputs involved in the first $(i+1)$ columns of $\calZ$ and by the first $i$ coordinates of $\Xi$. The first two conditions hold naturally by the definition of the filtration. For the third condition, we note that the $i$-th coordinate of $\Xi$ is independent of the first $i$ columns of $\calZ$ (by the previous paragraph) and the first $i-1$ coordinates of $\Xi$ (by the definition of $\Xi$), which generate $\calF_{i-1}$. Then condition 3) holds since the $i$-th coordinate of $\Xi$ has marginal distribution $\calN(0,1)$.

    Now we can apply \cite[Theorem 1]{abbasi2011improved} with $V=\frac{3\sigma_u^2\psi_A^2\psi_\Sigma^2}{2(1-\rho_A)^2} \cdot \rho_A^{2d} s_{\max}(T-t_0+1) \cdot I$ and get the following with probability at least $1-\delta/(2\cdot 5^{r})$
    \begingroup\makeatletter\def\f@size{9}\check@mathfonts\makeatother
    \begin{equation}\begin{split}\label{eq:cross_conc_3}
        \norm{\calZ\Xi}^2\overset{(i)}{\leq} {}& \sigma_{\max}\b{V+\calZ\calZ\t}\\
        {}& \cdot 2\log\b{\frac{\det\b{V+\calZ\calZ\t}^{1/2}\det\b{V}^{-1/2}}{\delta/10^r}}\\
    \end{split}\end{equation}
    \endgroup
    Here $(i)$ is by \cite[Theorem 1]{abbasi2011improved} with $\eta_s$ being the $s$-th coordinate of $\Xi$ and $X_s$ being the $s$-th column of $\calZ$.

    \textbf{Step 2: We now upper bound $\bm{\norm{\calZ}}$.} 
    By its definition,
    \begin{equation}\begin{split}\label{eq:cross_conc_21}
        {}& \calZ\calZ\t = \sum_{k\in[K]} \calU_k\calT_k\calT_k\t \calU_k\t\\
        = {}& \sum_{k\in[K]} \calU_k\bbI\b{v\t C_k\t C_k}\calT\calT\t  \bbI\b{C_k\t C_kv}\calU_k\t
    \end{split}\end{equation}
    For simplicity, we let $\calI_k = \{i: [C_k]^i \neq 0\}$. Namely, $\calI_k$ denotes the columns of $C_k$ that are non-zero, or intuitively, the state coordinates that are measured in trajectory $k$. Then $v\t C_k\t = [v\t]^{\calI_k}[C_k\t]_{\calI_k}$. 
    Therefore,
    \begingroup\makeatletter\def\f@size{8}\check@mathfonts\makeatother
    \begin{equation*}\begin{split}
        \calZ\calZ\t = {}& \sum_{k\in[K]} \calU_k\bbI\b{[v\t]^{\calI_k}[C_k\t]_{\calI_k} C_k}\b{\calT\calT\t}  \bbI\b{C_k\t[C_k]^{\calI_k}[v]_{\calI_k}}\calU_k\t\\
        \preceq {}& \norm{\calT\calT\t}\sum_{k\in[K]} \calU_k\bbI\b{[v\t]^{\calI_k}[C_k\t]_{\calI_k} C_kC_k\t[C_k]^{\calI_k}[v]_{\calI_k}}\calU_k\t\\
        = {}& \norm{\calT\calT\t}\sum_{k\in[K]} \calU_k\bbI\b{[v\t]^{\calI_k}[v]_{\calI_k}}\calU_k\t\\
        = {}& \norm{\calT\calT\t}\sum_{k\in[K]} \norm{[v]_{\calI_k}}^2\calU_k\calU_k\t.
    \end{split}\end{equation*}
    \endgroup
    Here the second last line is because $C_kC_k\t = I$, and because $[C_k\t]_{\calI_k}[C_k]^{\calI_k} = I$ for any measurement matrix $C_k$.

    Now taking the norm of both sides gives
    {\begin{equation}\begin{split}\label{eq:cross_conc_11}
        \norm{\calZ\calZ\t} \leq {}& \norm{\calT\calT\t }\sum_{k\in[K]} \norm{[v]_{\calI_k}}^2\norm{\calU_k\calU_k\t}.    
    \end{split}\end{equation}}

    When \Cref{eq:cross_conc_4} holds, we know by \Cref{lem:input_conc} that $\norm{\calU_k\calU_k\t} = \norm{\sum_{t=t_0}^TU_{k,t}(d)U_{k,t}(d)\t} \leq 3\sigma_u^2 (T-t_0+1)/2$ for every $k\in[K]$ with probability at least $1-\delta/(2\cdot 5^r)$. Substituting back gives the following with probability at least $1-\delta/(2\cdot 5^r)$
    \begin{equation}\begin{split}
        \norm{\calZ\calZ\t} \leq {}& \frac{3\sigma_u^2}{2} \norm{\calT\calT\t } (T-t_0+1) \cdot \sum_{k\in[K]} \norm{[v]_{\calI_k}}^2\\
    \end{split}\end{equation}
    Since every coordinate is measured at most $s_{\max}$ times in the $K$ trajectories, we know that $\sum_{k\in[K]} \norm{[v]_{\calI_k}}^2 \leq s_{\max}\norm{v}^2 = s_{\max}$. Therefore, with probability at least $1-\delta/(2\cdot 5^r)$,
    \begin{equation}\begin{split}\label{eq:cross_conc_5}
        \norm{\calZ\calZ\t} \leq \frac{3\sigma_u^2}{2} \norm{\calT\calT\t } \cdot s_{\max}(T-t_0+1).
    \end{split}\end{equation}

    Finally, we apply \cite[Lemma E.1]{proof_2tsiamis2019finite} on $\calT$, which bounds the norm of such Toeplitz matrices
    \begin{equation*}\begin{split}
        \norm{\calT} \leq {}& \sum_{\tau=0}^{T-t_0}\norm{A^{d+1+\tau}\Sigma} \leq \norm{\Sigma}\sum_{\tau=0}^{\infty}\psi_A\rho_A^{d+\tau}\leq \frac{\psi_A\psi_\Sigma}{1-\rho_A}\rho_A^d.
    \end{split}\end{equation*}
    Here the second inequality is by Assumption \ref{assmp:stable} and $\norm{\Sigma}\leq \psi_\Sigma$.
    Substituting back into \Cref{eq:cross_conc_5} and we get the following with probability at least $1-\delta/(2\cdot 5^r)$
    \begin{equation}\begin{split}\label{eq:cross_conc_10}
        \norm{\calZ}^2 \leq \frac{3\sigma_u^2\psi_A^2\psi_\Sigma^2}{2(1-\rho_A)^2} \rho_A^{2d}s_{\max}(T-t_0+1).
    \end{split}\end{equation}

    \textbf{Step 3.} With a union bound, combining \Cref{eq:cross_conc_3,eq:cross_conc_10} gives the following with probability at least $1-\delta/5^r$
    \begingroup\makeatletter\def\f@size{9}\check@mathfonts\makeatother
    \begin{equation*}\begin{split}
        {}& \norm{\calZ\Xi}^2 = \norm{\sum_{k=1,t=t_0}^{K,T} v\t C_k\t C_k \underline{G}_{t-t_0} \Xi_{k,t-t_0}(t-t_0) U_{k,t}(d)\t}^2\\
        \leq {}& \sigma_{\max}\b{V+\calZ\calZ\t} \cdot 2\log\b{\frac{\det\b{V+\calZ\calZ\t}^{\frac{1}{2}}\det\b{V}^{-\frac{1}{2}}}{\delta/10^r}}\\
        \leq {}&\sigma_{\max}\b{V+\calZ\calZ\t}\cdot 2\log\b{10^r\cdot 2^{m(d+1)/2}\cdot \frac{1}{\delta}}\\
        \leq {}& \frac{6\sigma_u^2\psi_A^2\psi_\Sigma^2}{(1-\rho_A)^2} \rho_A^{2d}s_{\max}(T-t_0+1)\cdot \max\left\{md,r\right\} \log\b{\frac{20}{\delta}}.
    \end{split}\end{equation*}
    \endgroup
    Here we use the definition $V=\frac{3\sigma_u^2\psi_A^2\psi_\Sigma^2}{2(1-\rho_A)^2} \cdot \rho_A^{2d} s_{\max}(T-t_0+1) \cdot I$.
    Finally, by \Cref{eq:cross_conc_9}, with probability at least $1-\delta$,
    \begingroup\makeatletter\def\f@size{9}\check@mathfonts\makeatother
    \begin{equation*}\begin{split}
        {}& \norm{\sum_{k=1,t=t_0}^{K,T} C_k\t C_k \underline{G}_{t-\tau} \Xi_{k,t-\tau}(t-\tau) U_{k,t}(d)\t}^2\\
        \leq {}& \frac{24\sigma_u^2\psi_A^2\psi_\Sigma^2}{(1-\rho_A)^2} \rho_A^{2d}\cdot \max\{md,r\}\cdot s_{\max}(T-t_0+1)\log\b{\frac{20}{\delta}}.
    \end{split}\end{equation*}
    \endgroup
\end{proof}

\begin{lemma}[Concentration of Crossterms with Fixed Dimensions]\label{lem:cross2_conc}  
    Consider system $(A,B,\sigma_u^2,\sigma_w^2)$ and measurement matrices $\{C_k\}_{k\in[K]}$ as in \Cref{thm:markov}. Consider trajectory length $T$ and estimation rank $d$ as in \Cref{alg:alloc}. Consider integer $t_0\in[0,d]$. Let $\{\xi_{k,t}\in\bbR^r\}_{k=1,t=d-t_0}^{K,T}$ be i.i.d. random vectors with $\xi_{k,t}\sim\calN(0,\sigma_{\xi}^2I_r)$. 
    Consider any $\delta\in(0,1)$. If $\xi_{k,t}$ is independent of $u_{k,\max\{t-d,0\}:T}$ and $u_{k',0:T}$ ($\forall k'\neq k$), and if $T \geq 2d + 16c\cdot rmd\cdot (\log^4(2mdT) + \log(10K/\delta))$ for absolute constant $c$ in \cite[Theorem 4.1, $L=1$]{krahmer2014suprema}, then the following holds with probability at least $1-\delta$,
    \begingroup\makeatletter\def\f@size{9}\check@mathfonts\makeatother
    \begin{equation*}\begin{split}
        {}& \norm{\sum_{k=1,t=d}^{K,T} C_k\t C_k \underline{G}_{t_0} \Xi_{k,t}(t_0) U_{k,t}(d)\t}\\
        \leq {}& \sqrt{\frac{24\sigma_u^2\sigma_{\xi}^2\psi_A^2}{\rho_A^2(1-\rho_A)^2}} \sqrt{\max\left\{md, r\right\} s_{\max}(T-d+1) \log\b{\frac{20}{\delta}}}.
    \end{split}\end{equation*}
    \endgroup
    Here we have defined $\underline{G}_{t_0} \coloneqq \begin{bmatrix}
        I & A & \cdots & A^{t_0}
    \end{bmatrix}$, $\Xi_{k,t}(t_0) \coloneqq \begin{bmatrix}
        \xi_{k,t}\t & \cdots & \xi_{k,t-t_0}\t
    \end{bmatrix}\t$. Moreover, recall $U_{k,t}(d) = \begin{bmatrix}
        u_{k,t}\t & \cdots & u_{k,t-d}\t
    \end{bmatrix}\t$.

\end{lemma}

\begin{proof}
    By \Cref{lem:gen_norm2Net}, there exists a unit vector $v\in\bbS^{r-1}$ s.t.
    \begingroup\makeatletter\def\f@size{8}\check@mathfonts\makeatother
\begin{equation}\begin{split}\label{eq:cross_conc_8}
        {}& \bbP\b{\norm{\sum_{k=1,t=d}^{K,T} C_k\t C_k \underline{G}_{t_0} \Xi_{k,t}(t_0) U_{k,t}(d)\t} > z}\\
        \leq {}& 5^r \bbP\b{\norm{\sum_{k=1,t=d}^{K,T} v\t C_k\t C_k \underline{G}_{t_0} \Xi_{k,t}(t_0) U_{k,t}(d)\t} > \frac{1}{2} z}.
    \end{split}\end{equation}    
    \endgroup

    \textbf{Step 1: We first upper bound $\norm{\bm{\sum_{k=1,t=d}^{K,T} v\t C_k\t C_k \underline{G}_{t_0} \Xi_{k,t}(t_0) U_{k,t}(d)\t}}$ with self-normalizing bound \cite[Theorem 1]{abbasi2011improved}. } For the unit vector $v\in\bbS^{r-1}$, define $g_{k,\tau}\t \coloneqq v\t C_k\t C_kA^{\tau}$, $\forall\tau\in[0,t_0]$. Then,
    \begin{equation}\begin{split}
        v\t C_k\t C_k\underline{G}_{t_0} = {}& v\t C_k\t C_k\begin{bmatrix}
            I & A & \cdots & A^{t_0}
        \end{bmatrix}\\
        = {}& \begin{bmatrix}
            g_{k,0}\t & g_{k,1}\t & \cdots & g_{k,t_0}\t
        \end{bmatrix}.
    \end{split}\end{equation}
    We then rewrite the objective as follows
    \begin{equation}\begin{split}\label{eq:cross_conc_14}
        {}& \sum_{k=1,t=d}^{K,T} v\t C_k\t C_k \underline{G}_{t_0} \Xi_{k,t}(t_0) U_{k,t}(d)\t\\
        = {}& \sum_{k=1,t=d}^{K,T} \b{\begin{bmatrix}
            g_{k,0}\t & \cdots & g_{k,t_0}\t
        \end{bmatrix} \begin{bmatrix}
            \xi_{k,t}\\
            \vdots\\
            \xi_{k,t-t_0}
        \end{bmatrix}} U_{k,t}(d)\t
    \end{split}\end{equation}
    Taking the transpose on both sides gives
    \begingroup\makeatletter\def\f@size{9}\check@mathfonts\makeatother
    \begin{align*}
        {}& \b{\sum_{k=1,t=d}^{K,T} v\t C_k\t C_k \underline{G}_{t_0} \Xi_{k,t}(t_0) U_{k,t}(d)\t}\t\\
        = {}& \sum_{k=1,t=d}^{K,T} U_{k,t}(d)\b{\sum_{\tau=0}^{t_0}g_{k,\tau}\t\xi_{k,t-\tau}}\\
        = {}& \sum_{k=1}^K\begin{bmatrix}
            U_{k,d}(d) & U_{k,d+1}(d) & \cdots & U_{k,T}(d)
        \end{bmatrix}\\
        {}& \begin{bmatrix}
            0 & \cdots & 0 & 0 & g_{k,0}\t & \cdots & g_{k,t_0}\t\\
            0 & \cdots & 0 & g_{k,0}\t & \cdots & g_{k,t_0}\t & 0\\
            \vdots & \iddots & \iddots & \iddots & \iddots & \iddots & \vdots\\
            0 & g_{k,0}\t & \cdots & g_{k,t_0}\t & 0 & \cdots & 0\\
            g_{k,0}\t & \cdots & g_{k,t_0}\t & 0 & \cdots & 0 & 0\\
        \end{bmatrix}\begin{bmatrix}
            \xi_{k,T}\\
            \xi_{k,T-1}\\
            \vdots\\
            \xi_{k,d-t_0}
        \end{bmatrix}\\
        = {}& \sum_{k=1}^K\underbrace{\begin{bmatrix}
            U_{k,d}(d) & \cdots & U_{k,T}(d)
        \end{bmatrix}}_{\calU_k} \bbI\b{v\t C_k\t C_k}\\
        {}& \underbrace{\begin{bmatrix}
            0 & \cdots & 0 & 0 & A^{0} & A^{1} & \cdots & A^{t_0}\\
            0 & \cdots & 0 & A^{0} & A^{1} & \cdots & A^{t_0} & 0\\
            \vdots & \iddots & \iddots & \iddots & \iddots & \iddots & \iddots & \vdots\\
            0 & A^{0} & A^{1} & \cdots & A^{t_0} & 0 & \cdots & 0\\
            A^{0} & A^{1} & \cdots & A^{t_0} & 0 & \cdots & 0 & 0\\
        \end{bmatrix}}_{\calT}\begin{bmatrix}
            \xi_{k,T}\\
            \xi_{k,T-1}\\
            \vdots\\
            \xi_{k,d-t_0}
        \end{bmatrix}.
    \end{align*}    
    \endgroup
    In the last line, we have used that $g_{k,\tau}\t = v\t C_k\t C_kA^{\tau}$ and that $\bbI(M) = \diag(\underbrace{M, M, \cdots, M}_{T-d+1})$. 
    We further define $\calT_k=\bbI\b{v\t C_k\t C_k}\calT$ and rewrite the objective as 
    \begin{equation*}\begin{split}
        {}& \b{\sum_{k=1,t=d}^{K,T} v\t C_k\t C_k \underline{G}_{t_0} \Xi_{k,t}(t_0) U_{k,t}(d)\t}\t\\
        = {}& \sum_{k=1}^K \calU_k\calT_k \Xi_{k,T}\b{T-d+t_0}\\
        = {}& \underbrace{\begin{bmatrix}
            \calU_1\calT_1 & \cdots & \calU_K\calT_K
        \end{bmatrix}}_{\calZ} \underbrace{\begin{bmatrix}
            \Xi_{1,T}\b{T-d+t_0}\\
            \vdots\\
            \Xi_{K,T}\b{T-d+t_0}\\
        \end{bmatrix}}_{\Xi}.
    \end{split}\end{equation*}

    We highlight that the $i$-th coordinate of $\Xi$ is independent of the first $i$ columns of $\calZ$. To see this, suppose the $i$-th coordinate of $\Xi$ is a coordinate of $\xi_{k,t}$ for some $k,t$. Then the $i$-th column of $\calZ$ must be composed of $U_{k,\max\{t,d\}}(d), \cdots, U_{k,\min\{t+t_0,T\}}(d)$. 
    Since $\xi_{k,t}$ is independent of $u_{k,\max\{t-d,0\}:T}$, we first conclude that \textit{$\xi_{k,t}$ is independent of the $i$-th column of $\calZ$}. Now observe that all columns $i'$ of $\calZ$ with $i'<i$ is only dependent on inputs $u_{k',t'}$ with either $k'\neq k$ or $k'=k \cap t'>\max\{t-d,0\}$. These inputs are all independent of $\xi_{k,t}$. Therefore, \textit{$\xi_{k,t}$ is independent of the first $i$ columns of $\calZ$}.
    
    Now we resort to the self-normalizing bound \cite[Theorem 1]{abbasi2011improved} to bound $\calZ\Xi$. To apply the theorem, we construct a filtration $\{\calF_i\}_{i=0}^{r(T-d+t_0+1)K}$ such that 1) the $i$-th column of $\calZ$ is $\calF_{i-1}$-measurable, 2) the $i$-th coordinate of $\Xi$ is $\calF_{i}$-measurable and 3) the $i$-th coordinate of $\Xi$ is subGaussian conditioned on $\calF_{i-1}$. To achieve this, we let $\calF_i$ be the sigma algebra generated by the inputs involved in the first $(i+1)$ columns of $\calZ$ and by the first $i$ coordinates of $\Xi$. The first two conditions hold naturally by the definition of the filtration. For the third condition, we note that the $i$-th coordinate of $\Xi$ is independent of the first $i$ columns of $\calZ$ (by the previous paragraph) and the first $i-1$ coordinates of $\Xi$ (by the definition of $\Xi$), which generate $\calF_{i-1}$. Then condition 3) holds since the $i$-th coordinate of $\Xi$ has marginal distribution $\calN(0,\sigma_{\xi}^2)$.

    Now we can apply \cite[Theorem 1]{abbasi2011improved} with $V=\frac{3\sigma_u^2\psi_A^2}{2\rho_A^2(1-\rho_A)^2}\cdot s_{\max}(T-d+1) I$ and get the following with probability at least $1-\delta/(2\cdot 5^r)$
    \begingroup\makeatletter\def\f@size{9}\check@mathfonts\makeatother
        \begin{equation}\begin{split}\label{eq:cross_conc_7}
        \norm{\calZ\Xi}^2
        \overset{(i)}{\leq} {}& \sigma_{\max}\b{V+\calZ\calZ\t}\\
        {}& \cdot 2\sigma_{\xi}^2\log\b{\frac{\det\b{V+\calZ\calZ\t}^{1/2}\det\b{V}^{-1/2}}{\delta/10^r}}\\
    \end{split}\end{equation}
    \endgroup

    Here $(i)$ is by \cite[Theorem 1]{abbasi2011improved} with $\eta_s$ being the $s$-th coordinate of $\Xi$ and $X_s$ being the $s$-th column of $\calZ$.
    
    \textbf{Step 2: We now upper bound $\bm{\norm{\calZ}}$.} 
    Following exactly the same reasoning as how we get \Cref{eq:cross_conc_5} from \Cref{eq:cross_conc_21}, the following holds with probability at least $1-\delta/(2\cdot 5^r)$,
    \begin{equation}\begin{split}\label{eq:cross_conc_12}
        \norm{\calZ\calZ\t} \leq \frac{3\sigma_u^2}{2} \norm{\calT\calT\t } \cdot s_{\max}(T-d+1).
    \end{split}\end{equation}

    Now we apply \cite[Lemma E.1]{proof_2tsiamis2019finite} on $\calT$, which bounds the norm of such Toeplitz matrices
    \begin{equation}\begin{split}
        \norm{\calT} \leq \sum_{\tau=0}^{t_0} \norm{A^{\tau}} \leq \sum_{\tau=0}^{t_0}\frac{\psi_A}{\rho_A}\rho_A^{\tau} = \frac{\psi_A}{\rho_A(1-\rho_A)}.
    \end{split}\end{equation}
    Here the second inequality is by Assumption \ref{assmp:stable}.
    Substituting into \Cref{eq:cross_conc_12} gives the following with probability at least $1-\delta/(2\cdot 5^r)$
    \begin{equation}\begin{split}\label{eq:cross_conc_22}
        \norm{\calZ}^2 \leq \frac{3\sigma_u^2\psi_A^2}{2\rho_A^2(1-\rho_A)^2} s_{\max}(T-d+1).
    \end{split}\end{equation}

    \textbf{Step 3.}
    With a union bound, combining \Cref{eq:cross_conc_7,eq:cross_conc_22} gives the following with probability at least $1-\delta/5^r$
    \begingroup\makeatletter\def\f@size{9}\check@mathfonts\makeatother
    \begin{equation*}\begin{split}
        {}& \norm{\calZ\Xi}^2 \leq \norm{\sum_{k=1,t=d}^{K,T} v\t C_k\t C_k \underline{G}_{t_0} \Xi_{k,t}(t_0) U_{k,t}(d)\t}^2\\
        \leq {}&\sigma_{\max}\b{V+\calZ\calZ\t}\cdot 2\sigma_{\xi}^2\log\b{\frac{\det\b{V+\calZ\calZ\t}^{\frac{1}{2}}\det\b{V}^{-\frac{1}{2}}}{\delta/10^r}}\\
        \leq {}&\sigma_{\max}\b{V+\calZ\calZ\t}\cdot 2\sigma_{\xi}^2\log\b{10^r\cdot 2^{m(d+1)/2}\cdot \frac{1}{\delta}}\\
        \leq {}& \frac{6\sigma_u^2\sigma_{\xi}^2\psi_A^2}{\rho_A^2(1-\rho_A)^2} s_{\max}(T-d+1) \cdot \max\left\{md, r\right\} \log\b{\frac{20}{\delta}}.
    \end{split}\end{equation*}
    \endgroup
    Here we have used $V=\frac{3\sigma_u^2\psi_A^2}{2\rho_A^2(1-\rho_A)^2} s_{\max}(T-d+1) I$.

    Finally, by \Cref{eq:cross_conc_8}, with probability at least $1-\delta$,
    \begingroup\makeatletter\def\f@size{9}\check@mathfonts\makeatother
    \begin{equation*}\begin{split}
        {}& \norm{\sum_{k=1,t=d}^{K,T} C_k\t C_k \underline{G}_{t_0} \Xi_{k,t}(t_0) U_{k,t}(d)\t}^2\\
        \leq {}& \frac{24\sigma_u^2\sigma_{\xi}^2\psi_A^2}{\rho_A^2(1-\rho_A)^2} \cdot \max\left\{md, r\right\} \cdot s_{\max}(T-d+1) \log\b{\frac{20}{\delta}}.
    \end{split}\end{equation*}
    \endgroup
\end{proof}

\vspace{-1em}\subsection{General Lemmas}
\begin{lemma}\label{lem:gen_norm2Net}
    For any random matrix $M\in\bbR^{a\times b}$, there exists a unit vector $x\in\mathbb{S}^{b-1}$ such that for all $\epsilon < 1$
    \begin{equation}\begin{split}
        \bbP(\norm{M} > z) \leq \b{1+\frac{2}{\epsilon}}^{b} \bbP\b{\norm{Mx} > (1-\epsilon)z}.
    \end{split}\end{equation}
    Here $\mathbb{S}^{b-1}$ denotes the unit sphere in $\mathbb{R}^b$.
\end{lemma}

\begin{proof}
    From Lemma 5.3 of \cite{proof_1vershynin2010introduction}, we know that
    \begin{equation}\begin{split}
        \bbP\b{\norm{M} > z} \leq \bbP\b{\max_{x\in\calN_{\epsilon}}\norm{Mx} > (1-\epsilon)z},
    \end{split}\end{equation}
    where $\calN_{\epsilon}$ is any $\epsilon$-net of $\bbS^{b-1}$ with $|\calN_{\epsilon}| \leq \b{1+2/\epsilon}^b$.
    On the other hand, 
    \begin{equation*}\begin{split}
        {}& \bbP\b{\max_{x\in\calN_{\epsilon}}\norm{Mx} > (1-\epsilon)z} \leq \sum_{x\in\calN_{\epsilon}} \bbP\b{\norm{Mx}>(1-\epsilon)z}.
    \end{split}\end{equation*}
    Combining the above inequalities gives
    \begin{equation}\begin{split}
        \bbP\b{\norm{M} > z} \leq \sum_{x\in\calN_{\epsilon}} \bbP\b{\norm{Mx}>(1-\epsilon)z}.
    \end{split}\end{equation}
    Thus, by the pigeonhole principle, there exists $x\in\calN_{\epsilon}$ s.t.
    \begin{equation}\begin{split}
        \bbP\b{\norm{M} > z} \leq {}& |\calN_{\epsilon}|\cdot \bbP\b{\norm{Mx}>(1-\epsilon)z}\\
        \leq {}& \b{1+\frac{2}{\epsilon}}^b\bbP\b{\norm{Mx}>(1-\epsilon)z}.
    \end{split}\end{equation}
\end{proof}

\begin{lemma}[Adapted from Theorem 4.1 of \cite{krahmer2014suprema}]\label{lem:circ}
    Consider any $\delta\in(0,1)$. Let $C\in\bbR^{L\times L}$ be a circulant matrix where the first row has distribution $\calN(0,I)$. For any $L_r\times L_c$ submatrix  $\Phi_{L_rL_c}$ of $C$, if 
    \begin{equation}\begin{split}
        L_r > cL_c\max\left\{\log^2(L_c)\log^2(L), \log(1/\delta)\right\},
    \end{split}\end{equation} 
    then with probability at least $1-\delta$
    \begin{equation}\begin{split}
        {}& \norm{\frac{1}{L_r}\Phi_{L_rL_c}\t\Phi_{L_rL_c}-I}\\
        \leq {}& \sqrt{c\frac{L_c}{L_r}\max\left\{\log^2(L_c)\log^2(L), \log(1/\delta)\right\}}.
    \end{split}\end{equation}
    Here $c$ is the absolute constant in \cite[Theorem 4.1, $L=1$]{krahmer2014suprema}.
\end{lemma}

\begin{proof}
    For the given circulant matrix $C$, consider any of its $L_r\times L$ submatrix $\Phi_{L_r}$. Consider the given constant $L_c \leq L$. Then the corresponding ``restricted isometry constant'' of $\frac{1}{\sqrt{L_r}}\Phi_{L_r}$, denoted by RIC, is defined as in \cite[Section 4]{krahmer2014suprema}
    \begin{equation*}\begin{split}
        \ric \coloneqq {}& \max_{x\in\bbR^L: \norm{x}=1 \text{ and }\norm{x}_0\leq L_c} \abs{\frac{1}{L_r}\norm{\Phi_{L_r}x}^2 - \norm{x}^2}\\
        = {}& \max_{x\in\bbR^L: \norm{x}=1 \text{ and }\norm{x}_0\leq L_c} \abs{x\t\b{\frac{1}{L_r}\Phi_{L_r}\t\Phi_{L_r}-I}x}.
    \end{split}\end{equation*}
    Consider any $L_r\times L_c$ submatrix $\Phi_{L_rL_c}$ of $\Phi_{L_r}$ and denote the indices of the $L_c$ columns in $\Phi_{L_r}\in\bbR^{L_r\times L}$ as $\calI$.
    \begin{equation*}\begin{split}
        \ric = {}& \max_{x\in\bbR^L: \norm{x}=1 \text{ and }\norm{x}_0\leq L_c} \abs{x\t\b{\frac{1}{L_r}\Phi_{L_r}\t\Phi_{L_r}-I}x}\\
        \overset{(i)}{\geq} {}& \max_{x\in\bbR^{L_c}: \norm{x}=1} \abs{x\t\b{\frac{1}{L_r}\Phi_{L_rL_c}\t\Phi_{L_rL_c}-I}x}\\
        = {}& \norm{\frac{1}{L_r}\Phi_{L_rL_c}\t\Phi_{L_rL_c}-I}.
    \end{split}\end{equation*}
    Here $(i)$ is because we restrict all possible vectors to vectors with non-zero entries only on rows whose indices are in $\calI$.

    On the other hand, by Theorem 4.1 of \cite{krahmer2014suprema}, the following holds with probability at least $1-\delta$ for any $\delta\in(0,1)$ if $L_r > cL_c\max\left\{\log^2(L_c)\log^2(L), \log(1/\delta)\right\}$
    \begin{equation}\begin{split}
        \ric \leq \sqrt{c\frac{L_c}{L_r}\max\{\log^2(L_c)\log^2(L), \log(1/\delta)\}}.
    \end{split}\end{equation}
    Here $c$ is the absolute constant in \cite[Theorem 4.1, $L=1$]{krahmer2014suprema}.

    Combining the above results finishes the proof.
\end{proof}

\section{Proofs for \Cref{sec:extension}}\label{proof:lem_param}
\begin{proof}[Proof of \Cref{thm:rank2}]
    \textbf{Step 1: We first prove \Cref{eq:sub4}.}
    Consider any $C\in\calC_{\calJ}$, i.e., any non-zero measurement matrix that only measures state coordinates in $\calJ$. We start by bounding the Hankel matrix perturbation. Each block row of $\htH(C)-H(C)$ is a submatrix of $[\htG-G]_{\calJ}$, and $H(C)$ has $r$ block rows. Therefore
    \begin{equation*}\begin{split}
        \norm{\htH(C)-H(C)} \leq {}& \sqrt{r}\,\max_{i\in[r]}\norm{[\htH(C)-H(C)]_i}\\
        \leq {}& \sqrt{r}\,\norm{[\htG-G]_{\calJ}} \leq \kappa_3\sqrt{\frac{s_{\max}\,r\,md}{s_{\min}^2 T}}.
    \end{split}\end{equation*}
    Let $c_1 \coloneqq \kappa_3^2 s_{\max} r md / s_{\min}$. Then $\norm{\htH(C)-H(C)} \leq \sqrt{c_1/(s_{\min}T)}$. By \Cref{eq:sub3}, $T > c_1^2/s_{\min}$, which is equivalent to $\sqrt{c_1/(s_{\min}T)} < \sqrt[4]{1/(s_{\min}T)}$. Therefore
    \begin{equation}\begin{split}\label{eq:rank2_pert}
        \norm{\htH(C) - H(C)} < \sqrt[4]{\frac{1}{s_{\min}T}}.
    \end{split}\end{equation}

    Now define $r_c \coloneqq \rank(H(C))$. By \Cref{eq:rank2_pert}, we have
    \begin{equation*}\begin{split}
        \sigma_{r_c}\b{\htH(C)} \geq {}& \sigma_{r_c}\b{H(C)} - \norm{\htH(C)-H(C)}\\
        \geq {}& \sigma_H - \sqrt[4]{\frac{1}{s_{\min}T}} > \sqrt[4]{\frac{1}{s_{\min}T}}.
    \end{split}\end{equation*}
    Here the first inequality uses $\sigma_{r_c}(H(C)) = \sigma_{\min}(H(C)) \geq \sigma_H$ by definition of $\sigma_H$, and the last inequality follows from \Cref{eq:sub3}, which gives $\sigma_H > 2\sqrt[4]{1/(s_{\min}T)}$.

    On the other hand, for any $i > r_c$ we have $\sigma_i(H(C)) = 0$. Therefore, by \Cref{eq:rank2_pert},
    \begin{equation*}\begin{split}
        \sigma_i\b{\htH(C)} \leq {}& \sigma_i\b{H(C)} + \norm{\htH(C)-H(C)} < \sqrt[4]{\frac{1}{s_{\min}T}}.
    \end{split}\end{equation*}

    By \Cref{alg:rank_unacc}, $\htr(C)$ is defined as the largest $i$ such that $\sigma_i\b{\htH(C)} > \sqrt[4]{1/(s_{\min}T)}$. By the above two inequalities, we know that $\htr(C) = r_c = \rank(H(C)) = \rank(O(C))$.

    \textbf{Step 2. We now prove \Cref{eq:sub6}.} This part of the proof is largely based on Theorem 1 in \cite{wolsey1982analysis}, or Section ``Submodular Optimization Algorithms'' in \cite{clark2017submodularity}. With coordinate set $\calI=\calJ$, \Cref{alg:2} is the greedy algorithm for the problem $\min_{\calS\subseteq \calJ} |\calS|$ such that $f(\calS) \geq r$,
    where $\calJ$ denotes the set of accessible state coordinates and $f\b{\calS} \coloneqq \rank\b{O\b{I_{\calS}}}$ with $I_{\calS} = [I_r]_{\calS}$.
    By Step 1, $\htr(I_{\calS}) = \rank\b{O\b{I_{\calS}}} = f\b{\calS}$ for every $\calS\subseteq \calJ$. This is equivalent to the greedy algorithm having exact access to $f(\cdot)$.

    By Theorem 7 of \cite{summers2015submodularity}, $f(\calS)$ is submodular and monotone increasing. This satisfies the condition of Theorem 1 in \cite{wolsey1982analysis}.
    Then by Theorem 1 in \cite{wolsey1982analysis}, the output $\htC\in\bbR^{\htn\times r}$ of \Cref{alg:2} satisfies $ \htn \leq \b{1+\log(r)} \tiln^*, \rank(O(\htC)) \geq r$.
    The second inequality implies $\rank(O(\htC)) = r$.
\end{proof}

\section*{References}
\bibliography{ref}

\end{document}